\documentclass[journal,12pt,letter,onecolumn]{IEEEtran}

\usepackage{cite}

\ifCLASSINFOpdf
  \usepackage[pdftex]{graphicx}
\else
  \usepackage[dvips]{graphicx}
\fi
\usepackage[cmex10]{amsmath}
\usepackage{algorithmic}

\usepackage{array}

\usepackage{mdwmath}

\usepackage{eqparbox}

\usepackage[tight,footnotesize]{subfigure}

\usepackage{fixltx2e}

\usepackage{stfloats}
\usepackage{url}

\usepackage{authblk}
\usepackage{setspace}
\usepackage{framed}
\usepackage{amssymb}
\usepackage{subfigure}
\usepackage{Color}

\newtheorem{theorem}{Theorem}
\newtheorem{lemma}[theorem]{Lemma}

\newtheorem{corollary}[theorem]{Corollary}

\newcommand{\CN}{\mathcal{N}}
\newcommand{\CJ}{\mathcal{J}}

\begin{document}
%
\title{Analysis of Cell Load Coupling for LTE Network Planning and Optimization}

%
%

\author[1]{Iana Siomina}
\author[1,2]{Di Yuan}
\affil[1]{{\small Ericsson Research, Ericsson AB, Sweden}}
\affil[2]{{\small Department of Science and Technology, Link{\"o}ping University, Sweden}}
\affil[ ]{\em{{\small Emails: iana.siomina@ericsson.com, diyua@itn.liu.se}}}



%
%


\markboth{Siomina and Yuan: Analysis of Cell Load Coupling for LTE
Network Planning and Optimization}{}

%



\maketitle

\begin{abstract}
  System-centric modeling and analysis are of key significance in
  planning and optimizing cellular networks. In this paper, we provide
  a mathematical analysis of performance modeling for LTE networks.
  The system model characterizes the coupling relation between the
  cell load factors, taking into account non-uniform traffic demand
  and interference between the cells with arbitrary network
  topology. Solving the model enables a network-wide performance
  evaluation in resource consumption.  We develop and prove both
  sufficient and necessary conditions for the feasibility of the
  load-coupling system, and provide results related to computational
  aspects for numerically approaching the solution.  The theoretical
  findings are accompanied with experimental results to instructively
  illustrate the application in optimizing LTE network configuration.
\end{abstract}

\begin{IEEEkeywords}
3.5G and 4G technologies, cell load coupling, network planning, optimization, system modeling
\end{IEEEkeywords}

%
\IEEEpeerreviewmaketitle

\setstretch{1.62}

\section{Introduction}
\label{sec:introduction}
%
%
%
%

Planning and optimization of LTE network deployment, such as base
station (BS) location and antenna parameter configuration, necessitate
modeling and algorithmic approaches for network-level performance
evaluation. Finding the optimal network design and configuration
amounts to solving an optimization problem of combinatorial nature.
Toward this end, system modeling admitting rapid performance
assessment in order to facilitate the selection among candidate
configuration solutions, of which the number is typically huge, is
essential. In this paper, we provide a rigorous analysis of an LTE
system performance model that works for general network topology and
explicitly accounts for non-uniform traffic demand. The performance
model that we study is referred to as the load-coupling system, to
emphasize the fact that the model characterizes the coupling relation
between the cells in their load factors. For each cell, the load
factor is defined as the amount of resource consumption in relation to
that is available in the cell. The load value grows with the cell's
traffic demand and the amount of inter-cell interference.
Intuitively, low load means that the network has more than enough
capacity to meet the demand, whilst high load indicates poor
performance in terms of congestion and potential service outage.  In
the latter case, the network design and configuration solution in
question should be revised, by reconfiguration or adding BS
infrastructure. Thus, simple means for evaluating the cell load for a
given candidate design solution is of high importance, particularly
because the evaluation may have to be conducted for a large number of
user demand and network configuration scenarios.

The load-coupling model for LTE networks takes the form of a
non-trivial system of non-linear equations. Calculating the solution
to the model, or determining solution existence, is not
straightforward. In this paper, we present contributions to
characterizing and solving the load-coupling system model. First, we
present a rigorous mathematical analysis of fundamental properties of
the system and its solution.  Second, we develop and prove a
sufficient and necessary condition for solution existence. Third, we
provide theoretical results that are important for numerically approaching the solution or delivering a bounding interval. Fourth, we instructively illustrate the application of the system model for
optimizing LTE network configuration.

The remainder of the paper is organized as follows. In Section
\ref{sec:related} we review some related works. The system model is
presented in Section \ref{sec:model}, and its fundamental properties
are discussed in Section \ref{sec:property}. In Section
\ref{sec:linear}, we present linear equation systems for the
purpose of determining solution existence.  In Section
\ref{sec:optimization}, we provide the relation between solving the
load-coupling system and convex optimization, and discuss approximate
solutions. The application of the system model and our theoretical
results to LTE network optimization is illustrated in Section
\ref{sec:experiment}, and conclusions are given in Section
\ref{sec:conclusion}.

\section{Related Works}
\label{sec:related}

Planning and performance optimization in cellular networks
form a very active line of research in wireless communications.
There are many works on UMTS network planning and
optimization. The research topics range from BS location and coverage planning
\cite{AmCaMa03,AmCaMa08,AmCaMa06,MaSc01,ZhYaAyWu07}, antenna parameter
configuration \cite{EiKoMaAcFuKoWeWe02,EiGeKoMaWe05,SiVaYu06},
to cell load balancing \cite{GaRuOl04,SiYu04}. For UMTS, the power
control mechanism that links together the cells in resource
consumption is an important aspect in performance modeling
\cite{Aein73,AmCaMa03,GrViGo92,Za92a,Za92b}. By power control, the transmit power of each link is
adjusted to meet a given signal-to-interference-and-noise ratio (SINR)
threshold. By the SINR requirement, the power expenditure of one cell
is a linear function in those of the other cells.  As a result, the
power control mechanism is represented by a system of linear
equations, which sometimes is referred to as UMTS interference
coupling
\cite{EiKoMaAcFuKoWeWe02,EiGeKoMaWe05}.
Interference coupling can be modeled for both downlink and
uplink. For network planning, the interference coupling system needs
to be solved many times for performance evaluation of different
candidate network configurations and multiple or aggregate user demand snapshots. In
\cite{MeHe01}, it is shown that, for both downlink
and uplink radio network planning, the dimension of the
power-control-based system of equations can be reduced from the number
of users in the system to the number of cells. The observation stems
from system characteristics that also form the foundation of
distributed power control mechanisms, see, e.g.,
\cite{GrViGo92,Za92b}. In
\cite{CaImMa04}, the authors provide theoretical properties of the
power-control-based system, and feasibility conditions in terms of
target data rates and QoS requirements. Motivated by the fact that
full-scale dynamic simulation is not computationally affordable for
large networks, the authors of \cite{ZhKoTiGo08} extend the UMTS
power-control system by a randomization-based procedure of service and
rate adaptation for HSUPA network planning.

In cellular network planning, the power-control
equation system is considered under given SINR threshold. Thus the
system solution and its existence are induced by the (candidate) network configuration in
question. In a more general context of wireless communications, power
control is often a means in performance optimization, that
is, the powers are optimization variables in minimizing or maximizing
objectives representing error probability, utility, QoS,
etc., that are all functions of SINR.  There is a vast amount of
theoretical analysis and algorithmic approaches for power optimization
under various (typically non-linear) objective functions, where a
gain matrix defines interference coupling
\cite{StFeWiBo10,StWiBo07,WiBoSt09,WiStBo08}.  In \cite{StWiBo07}, the
authors identify objective functions admitting a convex formulation of
power optimization, and develop a distributed
gradient-projection-based algorithm. Further developments include
algorithmic design utilizing Kuhn-Tucker condition \cite{WiBoSt09},
conditional Newton iteration yielding quadratic convergence
\cite{WiStBo08}, and model extension to include explicit
SINR-threshold constraints \cite{StFeWiBo10}.

Another line of research of power control is the characterization of
the achievable performance region under various utility functions and
interference functions. The authors of \cite{BoSt08} show the strict
convexity of the region for logarithmic functions of SINR.  In
\cite{BoNaAl11}, the authors characterize utility functions and
function transformation of power, for which the resulting power
optimization problem is convex. The investigation in \cite{BoNaSc11b}
provides conditions under which the boundary points of the region are
Pareto-optimal. In \cite{BoNaSc11a}, the authors present graph
representations of power and interference, and study the relation
between graph structure, irreducibility of the interference coupling
matrix, and the convexity of the utility region.

In contrast to the power-control model, the service requirement of
rate-control scheme in cellular networks is not SINR threshold (or a
function of SINR), but the amount of data to be served over a given
time period. Among other advantages, this approach makes it possible
to capture the effect of scheduling without the need of explicitly
modeling full details of scheduling algorithms. The rate-control-based
approach is primarily targeting, although not limited to,
non-power-controlled systems or systems with a fixed-rate traffic
demand. The approach has been less studied, but is of a high interest
for OFDMA-based networks. In general, the rate-control scheme exhibits
non-linear relations between the cell-coupling elements (in our case,
cell loads). The resulting model is therefore more complex than the
power-control model for UMTS. For power control, fundamental solution
characterizations are well-established for linear as well as more
general interference functions. For the latter, see, for example,
\cite{BoSc07}. For rate-control-based coupling systems (see
\cite{MaKo10} and Section \ref{sec:model}), a structural
difference from power control is that, in the former, one element cannot be expressed as a sum of terms, each being a function denoting the
impact of another element, and the coupling is not scale
invariant. For network planning, one known approach is
to consider an approximate linear function, obtained from
system-specific adaptive modulation and coding (AMC) parameters, to
represent the relation between date rate and SINR \cite{MaTuHuBo07},
and thereby arrive at a equation system being similar to that of UMTS.

From an engineering standpoint, LTE network optimization is becoming
increasingly important. In \cite{DaPaSkBe07}, the authors provide the
fundamental principles of LTE network operation and radio resource
allocation. Among the optimization issues, the research theme of
scheduling strategies and radio resource management (RRM) algorithms
has been extensively investigated. See, for example,
\cite{AsMo08,JaLe03,LeFaZhYa07,LePeMeXuLu09,PaHwCh07,PoKoMo06,ShAnEv05}
and the references therein.  Two major aspects considered in the
references are the balance between resource efficiency and fairness,
and quality of service awareness.  In \cite{Ek09}, the author gives a
survey of tools enabling service and subscriber differentiation.  For
cell planning, the propagation modeling, link budget consideration,
and performance parameters have been investigated in
\cite{SoSh10}.

High-level and accurate performance modeling is of high value in
planning cellular networks, as full-scale dynamic simulations are not
affordable for large planning scenarios (e.g., \cite{ZhKoTiGo08}).
The LTE system model that we analyze has been introduced by Siomina et
al.\ \cite{SiFuFo09} for studying OFDM network capacity region with
QoS consideration. The work in \cite{SiFuFo09} does not, however,
provide a general analysis of the model, and the major part of the
study relies on a simplification assuming uniform
load among cells. In the forthcoming sections, we present both analytical and
numerical results without these limitations.

Recently, the authors of \cite{MaKo10}
have presented a non-linear LTE performance model being very similar to the one
studied in the current paper.
That our performance model has been independently
proposed by others supports the modeling approach. The work in
\cite{MaKo10} provides further an approximation of load coupling via another
non-linear but simpler equation system, along with incorporating
continuous user distribution. Our study differs from \cite{MaKo10}, as
the focus of the current paper is a detailed investigation of
key properties and solution characterization of the load coupling
system.

\section{The System Model}
\label{sec:model}

Denote by $\CN = \{1, \dots, n\}$ the set of cells in a given network
design solution.  Without loss of generality, we assume that each cell
has one antenna to simplify notation.  The service area is represented
by a grid of pixels or small areas, each being characterized by
uniform signal propagation conditions. The set of pixels is denoted by
$\CJ$. The total power gain between antenna $i$ and pixel $j$ is
denoted by $g_{ij}$. We use $\CJ_i \subset \CJ$ to denote the serving
area of cell $i$. In a network planning context, both the gain matrix
as well as the cells' serving areas are determined by BS location and
antenna configuration.

For realistic network planning scenarios, the traffic demand is
irregularly distributed.  Let the user demand in pixel $j$ be denoted by
$d_j$. The demand represents the amount of data to be delivered to the users
located in pixel $j$ within the time interval under consideration.  By
defining a service-specific index, the demand parameter and the system
model can be extended to multiple types of services (see
\cite{SiFuFo09}). We will, however, consider one service type merely
for the sake of compactness.

We use $\rho_i$ to denote the level of resource consumption in
cell $i$. The entity is also referred to as cell load. In LTE systems,
the cell load can be interpreted as the expected fraction of the
time-frequency resources that are scheduled to deliver data.
The network-wise load vector, ${\boldsymbol \rho} = (\rho_1, \rho_2, \dots,
\rho_{n})^T$, plays a key role in performance modeling. In particular,
a well-designed network shall be able to meet the target demand
scenarios without overloading the cells. Hence the load vector forms a
natural performance metric in network configuration (cf.\ power
consumption in UMTS networks).  The load of a cell is a result of the
user demands in the pixels in the cell serving area, the channel
conditions, as well as the amount of interference. The last aspect
inter-connects the elements in the load vector, as the load of a cell
is determined by the SINRs and the resulting bit rates over the cell's
serving area, and these values are in turn dependent on the load
values of the other cells.  To derive the performance model, we
consider the SINR in pixel $j \in \CJ_i$ defined as follows,

\begin{equation}
\label{eq:sinr}
\gamma_j({\boldsymbol \rho}) = \frac{P_i g_{ij}}{\displaystyle \sum_{k \in \CN \setminus \{i\}} P_k g_{kj} \rho_k + \sigma^2}.
\end{equation}

In \eqref{eq:sinr}, $P_i$ is the power spectral density per minimum
resource unit in scheduling (in LTE, this corresponds to a pair of
time-consecutive resource blocks), and $\sigma^2$ is the noise
power. By \eqref{eq:sinr}, the inter-cell interference grows by the
load factor. In effect, $\rho_k$ can be interpreted as the probability
of receiving interference originating from cell $k$ on all the
sub-carriers of the resource unit. 
Let $B \log_2(1 +\gamma_j({\boldsymbol \rho}))$ be a function describing the effective bitrate per resource unit.
This formula is shown to be very
accurate for LTE downlink \cite{MoNaKoFrPoPeKoHuKu08}.  Thus to serve
demand $d_j$ in $j$, $\frac{d_j}{B \log_2(1 +
\gamma_j({\boldsymbol \rho}))}$ resource units are required.


Let $K$ denote the total number of resource units in the frequency-time domain in question,
and denote by $\rho_{ij}$ the proportion of resource consumption of cell $i$ due
to serving the users in $j \in \CJ_i$. By these definitions, we obtain
the following equation,
\begin{equation}
\label{eq:rhoij}
K \rho_{ij} = \frac{d_j} {B \log_2(1 +
\gamma_j({\boldsymbol \rho}))}.
\end{equation}

From \eqref{eq:rhoij}, it is clear that the load of a cell is a
function of the load levels of other cells.
Observing that $\rho_i = \sum_{j \in \CJ_i} \rho_{ij}$ and putting the
equations together lead to the following equation,

\vspace{-1mm}
\begin{equation}
\label{eq:rhoi}
\rho_i = \hspace{-1mm}\sum_{j \in \CJ_i} \rho_{ij}
= \hspace{-1mm}\sum_{j \in \CJ_i}  \frac{d_j} {KB \log_2(1 + \gamma_j({\boldsymbol \rho}))}
= \hspace{-1mm}\sum_{j \in \CJ_i}  \frac{d_j} {KB \log_2\left(1 +  \frac{P_i g_{ij}}{\sum_{k \in \CN \setminus \{i\}} P_k g_{kj} \rho_k + \sigma^2}\right)}~.
\end{equation}

The equation above represents the coupling relation between cells
in their resource consumption. In vector form, we have ${\boldsymbol
\rho} = {\boldsymbol f}({\boldsymbol \rho}, {\boldsymbol g}, {\boldsymbol d}, K,
B)$, where ${\boldsymbol f} = (f_1, \dots, f_i, \dots, f_n)^T$, and $f_i, i=1, \dots, n$,
represents the $\mathbb{R}_+^{n-1} \rightarrow
\mathbb{R}_+$ function as defined by \eqref{eq:rhoi}; here, $\mathbb{R}_+$ and $\mathbb{R}_+^{n-1}$ are used to denote the single- and $(n-1)$-dimension space of all real non-negative numbers, respectively.  Since in the
subsequent discussions there will be no ambiguity in the input
parameters, we use the following compact notation to denote the
non-linear equation system,
\begin{equation}
\label{eq:rho}
{\boldsymbol \rho} =  {\boldsymbol f} ({\boldsymbol \rho}).
\end{equation}

From \eqref{eq:rhoi}, three immediate observations follow.  First, for
all $i=1, \dots, n$, the load function $f_i$ is strictly increasing
in the load of other cells. Second, for non-zero $\sigma^2$, this function is
strictly positive when the load values of other cells (and thus
interference) are all zeros, i.e., $ {\boldsymbol f}(0) > 0$. Third, the function is continuous, and at
least twice differentiable for ${\boldsymbol \rho} \geq 0$.

From the network performance standpoint, the capacity is sufficient to
support the traffic demand, if equation system \eqref{eq:rho}
admits a load vector ${\boldsymbol \rho}$ with $0\leq \rho_i \leq 1, i \in
\CN$. In our analysis, however, we do not restrict ${\boldsymbol \rho}$ to be at most one, in
order to avoid any loss of generality. In addition, even if the
solution contains elements being greater than one, the values are of
significance in network planning, because they carry information of
the amount of shortage of resource in relation to the demand.

Solving \eqref{eq:rho} deals with finding a fixed point (aka invariant
point) of function ${\boldsymbol f}$ in $\mathbb{R}^n_+$, or
determines that such a point does not exist.  In the remainder of the
paper, we use $S$ as a general notation for the space of non-negative
solutions (fixed points) to systems of equations or inequalities. The
system in question is identified using subscript. Thus,
$S_{{\boldsymbol \rho} = {\boldsymbol f}({\boldsymbol
\rho})}$ denotes the solution space of \eqref{eq:rho}.  Note that, for
\eqref{eq:rho} as well as the linear equation systems to be introduced
later, only non-negative solutions are of interest.  Hence, throughout
the article, a (linear or non-linear) system is said to be feasible,
if there exists a solution for which non-negativity holds, otherwise
the system is said to be infeasible (even if a solution of negative
values exists).  The case that \eqref{eq:rho} is infeasible is denoted
by $S_{{\boldsymbol \rho} = {\boldsymbol f}({\boldsymbol \rho})} =
\emptyset$.

A useful optimization formulation in our analysis is the minimization
of the total cell load, subject to the inequality form of
\eqref{eq:rho}. The formulation is given below.
\begin{subequations}
\label{eq:optrho}
\begin{align}
\min~~ & \sum_{i \in {\cal N}} \rho_i \label{eq:optrhoobj}\\
& {\boldsymbol \rho} \geq  {\boldsymbol f}({\boldsymbol \rho}) \label{eq:optrhocons} \\
& {\boldsymbol \rho} \in  \mathbb{R}_+^n
\end{align}
\end{subequations}

For \eqref{eq:optrho}, its solution space $S_{{\boldsymbol \rho} \geq
{\boldsymbol f}({\boldsymbol \rho})}$ is also referred to as the
feasible load region.
Recall that ${\boldsymbol f}({\boldsymbol \rho})$
is strictly increasing, hence if $S_{{\boldsymbol
\rho} \geq {\boldsymbol f}({\boldsymbol \rho})} \not= \emptyset$, then
for any optimal solution to \eqref{eq:optrho},
\eqref{eq:optrhocons} holds with equality, as otherwise \eqref{eq:optrhoobj}
can be improved, contradicting that the solution is optimal.
In conclusion, any optimum of \eqref{eq:optrho} is a solution to
\eqref{eq:rho}.


We end the section by an illustration of the load-coupling
system for two cells in Figure \ref{fig:twocell}. The two cells have
symmetric parameters. In the figure, the two non-linear functions are
given by the solid lines. In the first two cases, system \eqref{eq:rho} has solutions
in ${\mathbb R}_+^2$,
though one of them represents a solution beyond the network capacity. In the last
case, the system is infeasible, as the two curves will never intersect
in the first quadrant. The straight lines with markers in the figure represent
linear equations related to \eqref{eq:rho}. Details of
these linear equations are deferred to Section \ref{sec:linear}.

\section{Fundamental Properties}
\label{sec:property}

In this section, we present and prove some fundamental properties of
the load-coupling system \eqref{eq:rho}. These theoretical results are
of key importance in the study of solution existence and
computation. For compactness, we introduce additional notation to
simplify \eqref{eq:rhoi} while keeping the essence of the
equation. Define $a_j = \frac{KB}{d_j}$,
$b_{ikj}=\frac{P_kg_{kj}}{P_ig_{ij}}$, and $c_{ij}=
\frac{\sigma^2}{P_i g_{ij}}$. These parameters contain,
respectively, the relation between the demand in pixel $j$ and
the resource in cell $i$, the inter-cell coupling in gain between cells
$k$ and $i$ in pixel $j$, and the channel quality of cell $i$ in relation to
noise in pixel $j$. The load equation \eqref{eq:rhoi} can then be written in
the following form,
\begin{equation}
\label{eq:rhois}
\rho_i = f_i({\boldsymbol \rho}), \textrm{~~where~~} f_i({\boldsymbol \rho}) = \sum_{j \in \CJ_i}  \frac{1} {a_{j} \log_2(1 +  \frac{1}{\sum_{k \in \CN \setminus \{i\}} b_{ikj} \rho_k + c_{ij}})}.
\end{equation}

The first fundamental property of \eqref{eq:rho} is how fast the load
of a cell asymptotically grows in the load of another cell.  We
formulate and prove the fact that, in the limit, the first-order
partial derivative of the load function converges to a constant.  For
any two cells $i, k$ ($i \not= k$), $\frac{\partial \rho_i}{\partial
\rho_k}$ is equal to

\vspace{-2mm}
\begin{equation}
\label{eq:partial}
\sum_{j \in \CJ_i} \ln(2) \frac{b_{ikj}}{a_j} \frac{1}{\ln^2 (1+\frac{1}{\sum_{h \in \CN \setminus \{i\}} b_{ihj} \rho_h + c_{ij}(\boldsymbol \rho)}) (\sum_{h \in \CN \setminus \{i\}} b_{ihj} \rho_h + c_{ij})^2(\boldsymbol \rho) (1 + \frac{1}{\sum_{h \in \CN \setminus \{i\}} b_{ihj} \rho_h + c_{ij}(\boldsymbol \rho)})}.
\end{equation}
\vspace{2mm}

\begin{theorem}
\label{theo:partiallimit}
$\displaystyle \lim_{\rho_k \rightarrow \infty} {\frac{\partial f_i}{\partial
\rho_k}} = \displaystyle \sum_{j \in \CJ_i} \ln(2) \frac{b_{ikj}}{a_j}$
\end{theorem}

\begin{proof}
Consider the component for pixel $j$ in the sum in \eqref{eq:partial},
and ignore the constant multiplier
$\ln(2)\frac{b_{ikj}}{a_{j}}$. Letting $u = \sum_{h \in \CN \setminus
\{i\}} b_{ihj} \rho_h + c_{ij}$, \eqref{eq:partial} can be written as
the following expression.
\begin{eqnarray}
\frac{1}{u^2 (1+\frac{1}{u}) \ln(1+\frac{1}{u}) \ln(1+\frac{1}{u})} & = & \frac{1}{\ln(1+\frac{1}{u})^u \ln(1+\frac{1}{u})^u +
 \ln(1+\frac{1}{u})^u \ln(1+\frac{1}{u})} \nonumber
\end{eqnarray}

The theorem follows then from the facts that $\lim_{u \rightarrow
\infty} (1+\frac{1}{u})^u = e$ and $u$ is linear in $\rho_k$.
\end{proof}

By Theorem \ref{theo:partiallimit}, the load of a cell increases
linearly in the load of another cell in the limit, i.e., the function
converges to a line in the high-load region. Moreover, the slope of
the line is strictly positive.

The next fundamental property is concavity.  The examples in Figure
\ref{fig:twocell} indicate that the load of a cell is a strictly
concave function in the other cell's load. We show that this is
generally true.

\begin{theorem}
\label{theo:concave}
For any cell $i \in \CN$, $f_i$ is strictly concave for $(\rho_i, \dots,
\rho_{i-1}, \rho_{i+1}, \dots, \rho_n) \in R^{n-1}_{+}$.
\end{theorem}

\begin{proof}
Without loss of generality, consider $f_n$ and its $(n-1)
\times (n-1)$ Hessian matrix. Let
$u = \sum_{h \in \CN \setminus \{n\}} b_{nhj} \rho_h + c_{nj}$. For
two cells $k$ and $h$, the Hessian element has the following
expression.
%

\begin{equation}
\label{eq:hessianelement}
\frac{\partial^2 f_n}{\partial \rho_k \partial \rho_h} =  \ln(2) \sum_{j \in \CJ_n} \frac{b_{nkj} b_{nhj}}{a_j} \cdot
\frac{\ln(1+\frac{1}{u}) [2 - (2u+1)\ln(1+\frac{1}{u})]}
{[\ln^2(1+\frac{1}{u})(u^2+u)]^2}
\end{equation}

Let $q(u)= 2- (2u+1)\ln(1+\frac{1}{u})$. We show that $q(u)<0$ for $u
>0$. This holds, for example, for $u=1$.
Next, $\lim_{u \rightarrow \infty} q(u) = 2 - \lim_{u \rightarrow \infty} \ln
\left[ (1+\frac{1}{u})^u (1+\frac{1}{u})^u (1+\frac{1}{u})\right]$
$=0$. Consider $q'(u)$ and $q''(u)$:
$q'(u) = -2 \ln\left(1 + \frac{1}{u}\right) + \frac{1}{u} + \frac{1}{u+1}$,
and  $q''(u) = \frac{-1}{u^2(u+1)^2} <0, \forall u > 0$.
%
Therefore $q'(u)$ is strictly decreasing and $\lim_{u \rightarrow \infty} q'(u) = 0$.
Hence, $q'(u) > 0$, meaning that $q(u)$ is strictly increasing for
$u > 0$. This, together with $q(1)<0$ and $\lim_{u
\rightarrow \infty} q(u) = 0$, prove that $q(u)<0$, $\forall u>0$.
By the definition of $u$, for ${\boldsymbol \rho} \geq 0$, $u \geq
c_{nj}$ which a strictly positive number for non-zero noise
power. Hence \eqref{eq:hessianelement} is well-defined and negative
for all ${\boldsymbol \rho} \geq 0$. Next, observe that
the Hessian matrix is the result of the following expression.
\begin{equation}
\label{eq:hessianreformulation}
\sum_{j \in \CJ_i} \frac{(b_{n1j}, \dots, b_{n(n-1)j}) (b_{n1j}, \dots, b_{n(n-1)j})^T}{a_j}
\cdot \frac{\ln(2) \ln(1+\frac{1}{u})q(u)}{[\ln^2(1+\frac{1}{u})(u^2+u)]^2}
\end{equation}

Because of the form of \eqref{eq:hessianreformulation} and that $q(u)<0,
\forall u \geq c_{nj}>0$, the Hessian matrix is negative
definite for any ${\boldsymbol \rho} \geq 0$. Hence the conclusion.
\end{proof}

From the concavity result, it follows that, for any cell $i$,
$f_i(\rho_1, \dots, \rho_{i-1}, \rho_{i+1}, \dots, \rho_n) -
\rho_i$ exhibits a strict radially quasiconcave structure.
A function is radially quasiconcave, if for a given stationary
positive point, which is in our case a solution to $f_i(\rho_1, \dots,
\rho_{i-1}, \rho_{i+1}, \dots, \rho_n) =
\rho_i$, and any scalar in range $(0,1)$, the
function value of the scaled point is greater than or equal to
zero. If the value is positive, the function is strictly radially
quasiconcave.

\begin{corollary}
\label{theo:radiallyquasiconcave}
For each $i \in \CN$, $f_i (\rho_1, \dots,
\rho_{i-1}, \rho_{i+1}, \dots, \rho_n) - \rho_i$ is strictly
radially quasiconcave, i.e., if
$f_i (\rho_1, \dots, \rho_{i-1}, \rho_{i+1}, \dots, \rho_n) = \rho_i$,
then $f_i (\lambda \rho_1, \dots, \lambda \rho_{i-1}, \lambda \rho_{i+1}, \dots, \lambda \rho_n) > \lambda \rho_i$ for any $\lambda \in (0,1)$.
\end{corollary}

\begin{proof}
Note that
$f_i (\lambda (\rho_1, \dots, \rho_{i-1}, \rho_{i+1}, \dots, \rho_n))$
$= f_i(\lambda (\rho_1, \dots, \rho_{i-1}, \rho_{i+1}, \dots, \rho_n) + 0 (1-\lambda)$.
By Theorem \ref{theo:concave}, we have
$f_i (\lambda (\rho_1, \dots, \rho_{i-1}, \rho_{i+1}, \dots, \rho_n)) > \lambda f_i(\rho_1, \dots, \rho_{i-1}, \rho_{i+1}, \dots, \rho_n) + (1-\lambda)f_i(0)$. Since $f_i ((\rho_1, \dots, \rho_{i-1}, \rho_{i+1}, \dots, \rho_n)) = \rho_i$ and $f_i(0)>0$, the result follows.
\end{proof}

In a real-life LTE network, if the capacity is sufficient to accommodate
the demand, then the network load will be at a stable working
point, which should be
unique. Thus the performance model ${\boldsymbol f}$ is reasonable only if
uniqueness holds mathematically. The following theorem states this is
indeed the case. In the rest of the paper, the unique solution,
if it exists, is denoted by ${\boldsymbol \rho}^*$.

\begin{theorem}
\label{theo:uniqueness}
If $S_{{\boldsymbol f}({\boldsymbol \rho})={\boldsymbol \rho}} \not= \emptyset$, then it is a singleton, i.e.,
${\boldsymbol f}({\boldsymbol \rho})={\boldsymbol \rho}$ has
at most one solution ${\boldsymbol \rho}^*$ in ${\mathbb R}_+^n$.
\end{theorem}

\begin{proof}
Suppose there are two solutions ${\boldsymbol \rho}^1$ and
${\boldsymbol \rho}^2$, both satisfying \eqref{eq:rho}, and
${\boldsymbol \rho}^1 \not= {\boldsymbol \rho}^2$.  Let $m \in
\text{argmin}_{i=1,\dots,n}
\rho^1_i / \rho^2_i$, and $\lambda = \rho^1_m / \rho^2_m$. Thus
$\rho^1_m = \lambda \rho^2_m$.  Assume $\lambda < 1$.  Then by
construction, $\lambda {\boldsymbol
\rho}^2 \leq {\boldsymbol \rho}^1$, and because ${\boldsymbol f}$ is strictly
increasing in the domain of $\mathbb{R}_+^n$, $f_m(\lambda {\boldsymbol
\rho}^2) \leq f_m({\boldsymbol \rho}^1)$.  Also, by Lemma
\ref{theo:radiallyquasiconcave}, $\lambda \rho^2_m < f_m(\lambda
{\boldsymbol \rho}^2)$, and thus $f_m({\boldsymbol \rho}^1) > \lambda
\rho_m^2$.  Note that $f_m({\boldsymbol \rho}^1) = \rho_m^1 = \lambda
\rho_m^2$ gives an contradiction. Therefore $\lambda > 1$. Considering
scaling down ${\boldsymbol \rho}^1$ with $\lambda$ instead, and
applying the same line of argument, a similar contradiction is
obtained. Hence the conclusion.
\end{proof}

\section{Determining Solution Existence and Lower Bounding}
\label{sec:linear}

Having proven solution uniqueness, we examine the existence of
${\boldsymbol \rho}^*$, that is, whether or not \eqref{eq:rho} has a
fixed point. There are a number of theorems characterizing the
existence of a fixed point (e.g., Brouwer's fixed-point theorem in
topology).  However, these results do not apply to \eqref{eq:rho}
because, in general, the output of function ${\boldsymbol f}$ is not
confined to a compact set in ${\mathbb R}^n_+$. In this section, we use a
linear equation system for analyzing solution existence. To this end,
we first present and prove some basic properties of the optimization
formulation \eqref{eq:optrho}.

\begin{theorem}
\label{theo:optsolution}
Assume $S_{{\boldsymbol \rho} \geq  {\boldsymbol f}({\boldsymbol \rho})} \not=
\emptyset$, i.e., there exists ${\boldsymbol {\bar \rho}} \geq
 {\boldsymbol f}({\boldsymbol {\bar \rho}}) > 0$, then \eqref{eq:optrho} has an optimal
solution.
\end{theorem}

\begin{proof}
Consider the optimization problem $\min \sum_{i \in {\cal N}}
\rho_i, {\boldsymbol \rho}\in {\bar S}$, where ${\bar
S} = S_{{\boldsymbol \rho} \geq  {\boldsymbol f}({\boldsymbol \rho})} \cap
\{{\boldsymbol \rho} \leq {\boldsymbol {\bar \rho}}\}$.
By the assumption in the theorem, ${\bar S} \not= \emptyset$.
From the definition of $\bar S$, it is clear that
any point being arbitrarily close to $\bar S$ (i.e.,
boundary point) is in the set, thus $\bar S$ is closed.
In addition, $\bar S$ is bounded since $\bar S \subseteq
\mathbb R^n_{+} \cap \{{\boldsymbol \rho} \leq {\boldsymbol {\bar \rho}}\}$.
Hence $\bar S$ is compact, and the result follows from Weierstrass
theorem in optimization.
\end{proof}

\begin{corollary}
\label{theo:opttoexistence}
If there exists ${\boldsymbol {\bar \rho}} \geq  {\boldsymbol f}({\boldsymbol {\bar
\rho}}) > 0$, then $S_{{\boldsymbol \rho} =  {\boldsymbol f}({\boldsymbol \rho})}
\not= \emptyset$.
\end{corollary}

\begin{proof}
Follows immediately from Theorem \ref{theo:optsolution} and the
previously made observation that any optimal solution to
\eqref{eq:optrho} satisfies \eqref{eq:optrhocons} with equality.
\end{proof}

To further characterize solution existence, we define the following type
of linear equation systems,
\begin{equation}
\label{eq:linear}
{\boldsymbol \rho} = {\boldsymbol h}({\boldsymbol \rho}) = {\boldsymbol H}\cdot({\boldsymbol \rho} - {\boldsymbol {\widehat \rho}}) + {\boldsymbol f(\widehat{\boldsymbol\rho})},
\end{equation}
where ${\boldsymbol h} = (h_1, \dots, h_n)$ is a vector of linear
functions, and each of them is defined in ${\mathbb R}^{n-1}_+
\rightarrow {\mathbb R}_+$, ${\boldsymbol {\widehat \rho}}$ is a
vector in ${\mathbb R}^n_+$ with given values, and ${\boldsymbol
f(\widehat{\boldsymbol\rho})}$ is a vector-function with elements
defined by \eqref{eq:rhois}. In \eqref{eq:linear}, ${\boldsymbol H}$
is an $n \times n$ matrix where the diagonal elements, $H_{ii},
i=1,\dots, n$ are zeros, and the other elements $H_{ik}, i \not=k$,
are strictly positive. Note that if ${\boldsymbol H}$ is the Jacobian of
function ${\boldsymbol f}$ evaluated at point ${\boldsymbol
{\widehat \rho}}$, \eqref{eq:linear} is a linearization of the
non-linear equation system \eqref{eq:rho} where the right-hand side of
\eqref{eq:linear} represents the tangent hyperplane to
function ${\boldsymbol f({\boldsymbol\rho})}$ at
$\widehat{\boldsymbol\rho}$. Such linear approximations are further
discussed in Section~\ref{sec:optimization}.

Observing the fact that the partial derivative
\eqref{eq:partial} asymptotically approaches a constant, as formulated
in Theorem
\ref{theo:partiallimit}, we consider linear
approximation of ${\boldsymbol f}$ by means of the linear function
having the limit values of the partial derivatives as the matrix
elements in ${\boldsymbol H}$, and passing through the point defined
by the load function values with zero load. Define ${\boldsymbol
h}^0$ the case of ${\boldsymbol h}$ where $H_{ik} = \ln(2)
\sum_{j \in \CJ_i} b_{ikj} / a_j$ for $k\neq i$, and ${\boldsymbol {\widehat \rho}}=0$.
For this linear approximation, there are similarities between the
elements of ${\boldsymbol H}$ and the UMTS interference-coupling
matrix (see, e.g., \cite{EiGeKoMaWe05,NaDo06}) in that both
capture the relation between gain factors of the serving and
interfering cells; however, the target QoS in the
interference-coupling matrix is link quality, whilst in ${\boldsymbol
H}$ it is given by the amount of user traffic demand.

If $S_{{\boldsymbol
\rho} = {\boldsymbol h}^0 ({\boldsymbol \rho})} \not= \emptyset$,
the solution, denoted by ${\boldsymbol \rho}_{\boldsymbol h}^0$, is
clearly unique.  The lemma below states that the linear function ${\boldsymbol h}^0$
provides an under-estimation of the true load function $\boldsymbol
f$, thus ${\boldsymbol \rho}_{\boldsymbol h}^0$, if exists,
gives a lower bound on the solution to the non-linear system \eqref{eq:rho}.

\begin{lemma}
\label{theo:under}
${\boldsymbol h}^0({\boldsymbol \rho}) \leq {\boldsymbol f}({\boldsymbol \rho})$
for any ${\boldsymbol \rho} \geq 0$.
\end{lemma}

\begin{proof}
We prove the validity of the result for an arbitrary cell $i$, that is,
$h_i^0({\boldsymbol \rho}) \leq f_i({\boldsymbol \rho})$,
${\boldsymbol \rho} \geq 0$. Because both
$h_i^0({\boldsymbol \rho})$ and $f_i({\boldsymbol \rho})$
are formed by a sum over $j \in \CJ_i$, it is sufficient to
establish the inequality for any $j \in \CJ_i$.
Let $u = \sum_{k \in \CN \setminus \{i\}} b_{ikj}\rho_k + c_{ij}$. The proof boils down
to showing the following inequality.

\vspace{-2mm}
$$
\frac{1} {\log_2(1+\frac{1}{u})} - (u-c_{ij}) \ln(2)
= \ln(2) \left(\frac{1} {\ln(1+\frac{1}{u})} - (u-c_{ij})\right)
\geq \frac{1}{\log_2(1+\frac{1}{c_{ij}})} = \frac{\ln(2)}{\ln(1+\frac{1}{c_{ij}})}
$$

Note that $u \geq c_{ij}$ by definition.
The inequality holds as equality for $u=c_{ij}$. It is then sufficient to prove that
$\frac{1}{\ln(1+\frac{1}{u})}-(u-c_{ij}) $ is increasing for $u \geq c_{ij}$.
Taking the derivative and doing some simple manipulations, one can conclude that the derivative
is non-negative corresponds to the inequality below.

\begin{equation}
\label{eq:qzone}
q(u) = u(u+1) \ln^2\left(1+\frac{1}{u}\right) \leq 1, ~~u \geq c_{ij}
\end{equation}

One can show easily that $\lim_{u \rightarrow 0^+} q(u) = 0$, hence
\eqref{eq:qzone} is satisfied for some $u \leq c_{ij}$.
Moreover, $\lim_{u \rightarrow \infty} q(u) = \lim_{u \rightarrow
\infty} \ln(1+\frac{1}{u})^u \ln(1+\frac{1}{u})^u +
\ln(1+\frac{1}{u})^u \ln(1+\frac{1}{u}) = 1$.
Hence it suffices to prove that $q'(u)
= (2u+1) \ln^2\left(1 + \frac{1}{u}\right) - 2\ln\left(1+\frac{1}{u}\right) \geq 0$, $u \geq 0$.
Using the fact that $\ln(1+\frac{1}{u}) >0$ for all $u > 0$,
the non-negativity of $q'(u)$ for $u \geq 0$ becomes equivalent to
that the second numerator in \eqref{eq:hessianelement} is negative,
which is proven in the proof of Theorem \ref{theo:concave}, and the
result follows.
\end{proof}

From Lemma \ref{theo:under}, one can expect that the load-coupling system
\eqref{eq:rho} has a solution, only if a solution exists to
${\boldsymbol \rho} = {\boldsymbol h}^0({\boldsymbol
\rho})$. The following theorem formalizes this necessary condition, and
establishes the result that ${\boldsymbol \rho}_{\boldsymbol h}^0$
bounds ${\boldsymbol \rho}^*$ from below.

\begin{theorem}
\label{theo:necessary}
If $S_{{\boldsymbol \rho} = {\boldsymbol f} ({\boldsymbol \rho})} \not= \emptyset$,
then $S_{{\boldsymbol \rho} = {\boldsymbol h}^0 ({\boldsymbol \rho})} \not= \emptyset$
and ${\boldsymbol \rho}_{\boldsymbol h}^0 \leq {\boldsymbol \rho}^*$.
\end{theorem}

\begin{proof}
Consider the following linear programming (LP) formulation.

\vspace{-4mm}
\begin{subequations}
\label{eq:linearrho}
\begin{align}
\min~~ & \sum_{i \in {\cal N}} \rho_i \label{eq:linearrhoobj}\\
& {\boldsymbol \rho} \geq  {\boldsymbol h}^0({\boldsymbol \rho}) \label{eq:linearrhocons} \\
& {\boldsymbol \rho} \in  \mathbb{R}_+^n
\end{align}
\end{subequations}

Similar to the result in Theorem \ref{theo:optsolution}, it can be
easily proven that \eqref{eq:linearrho} has an optimal solution if there
exists any ${\boldsymbol \rho} \geq 0$ satisfying
\eqref{eq:linearrhocons}. In addition,
it is clear that any optimum is in $S_{{\boldsymbol \rho} =
{\boldsymbol h}^0({\boldsymbol \rho})}$, and $S_{{\boldsymbol \rho} =
{\boldsymbol h}^0({\boldsymbol \rho})}$ is either empty or a
singleton.  Consider ${\boldsymbol
\rho}^*$. By Lemma \ref{theo:under}, ${\boldsymbol h}^0({\boldsymbol
\rho}^*) \leq {\boldsymbol f}({\boldsymbol \rho}^*) = {\boldsymbol
\rho}^*$.  Hence ${\boldsymbol \rho}^*$ is a feasible solution
to \eqref{eq:linearrho}. It follows then $S_{{\boldsymbol \rho} =
{\boldsymbol h}^0({\boldsymbol \rho})} \not= \emptyset$.
Furthermore, \eqref{eq:linearrho} obviously remains feasible with the
additional constraint ${\boldsymbol \rho} \leq {\boldsymbol \rho}^*$.
Since the LP optimum is unique and equal to ${\boldsymbol
\rho}_{\boldsymbol h}^0$, ${\boldsymbol \rho}_{\boldsymbol h}^0 \leq
{\boldsymbol \rho}^*$.
\end{proof}

By Theorem \ref{theo:necessary}, the linear system ${\boldsymbol \rho}
= {\boldsymbol h}^0 ({\boldsymbol \rho})$ is potentially useful for
detecting infeasibility.  If the linear system is infeasible, then it
is not meaningful to attempt to solve \eqref{eq:rho}.  In addition, if
feasibility holds for
${\boldsymbol \rho} = {\boldsymbol h}^0 ({\boldsymbol \rho})$, the solution provides a
lower bound to the true load values. Thus having ${\boldsymbol
\rho}_{\boldsymbol h}^0$ close to one indicates an overloaded network,
and its corresponding configuration can be discarded from further
consideration in network planning, without the need of solving the
non-linear system \eqref{eq:rho}.

In Figure \ref{fig:twocell}, the
lines with markers represent the linear function ${\boldsymbol h}^0$.  In the first
two cases, ${\boldsymbol \rho}^*$ exists, and solving the linear system leads to a
lower bound ${\boldsymbol \rho}_{\boldsymbol h}^0$ (i.e., the intersection point of the lines) of
${\boldsymbol \rho}^*$. In the last case, the linear system has no solution, and
consequently $S_{{\boldsymbol \rho} = {\boldsymbol f} ({\boldsymbol
\rho})} = \emptyset$.

Thus far, it has become clear that ${\boldsymbol \rho} =
{\boldsymbol h}^0 ({\boldsymbol \rho})$ provides an optimistic view
of the cell load. We are able to prove a slightly unexpected but much stronger
result. The linear equations ${\boldsymbol \rho} = {\boldsymbol h}^0 ({\boldsymbol
\rho})$, in fact, give an exact characterization of solution existence of the
load-coupling system. Namely, that ${\boldsymbol \rho} = {\boldsymbol
h}^0 ({\boldsymbol \rho})$ has a solution is not only a necessary, but
also a sufficient condition for the feasibility of \eqref{eq:rho}.

The intuition of the sufficiency result is as follows. Consider Figures
\ref{fig:case1}-\ref{fig:case2}, for which the linear equation system
has solution. Suppose the slopes of the lines are increased slightly.
Intuitively, if the increase is sufficiently small, the new linear
system will remain feasible. Also, the figure gives the hint that the
modified linear function will eventually go above the non-linear load
function for large load, indicating  $S_{{\boldsymbol \rho} = {\boldsymbol f} ({\boldsymbol \rho})}
\not= \emptyset$. To rigorously prove the result, we define the linear
equation system ${\boldsymbol \rho} = {\boldsymbol
h}^{\epsilon}({\boldsymbol \rho})$, obtained by increasing the slope
coefficients of ${\boldsymbol h}^0$ by a positive constant
$\epsilon$. That is, ${\boldsymbol h}^{\epsilon}$ denotes the case of
\eqref{eq:linear} where $H_{ik} = \ln(2) \sum_{j
\in
\CJ_i} b_{ikj} / a_j + \epsilon$, ${\boldsymbol {\widehat \rho}}=0$.

\begin{lemma}
\label{theo:epsilonexistence}
If $S_{{\boldsymbol \rho} = {\boldsymbol h}^0 ({\boldsymbol \rho})}
\not= \emptyset$, i.e., ${\boldsymbol \rho}_{\boldsymbol h}^0$ exists,
then there exists $\epsilon > 0$ such that
$S_{{\boldsymbol \rho} = {\boldsymbol h}^{\epsilon} ({\boldsymbol
\rho})} \not=
\emptyset$.
\end{lemma}

\begin{proof}
First, note that $S_{{\boldsymbol \rho} \geq {\boldsymbol h}^0 ({\boldsymbol
\rho})}$ has a non-empty interior. In particular, it is
easily verified that $\lambda {\boldsymbol \rho}_{\boldsymbol h}^0$ is
an interior point for any $\lambda > 1$.  Denote by ${\boldsymbol
{\tilde
\rho}}$ such a point, that is, ${\tilde \rho}_i >
h_i^0({\boldsymbol {\tilde \rho}}), i\in \CN$.
Letting $\epsilon_i = \frac{{\tilde \rho}_i - h_i^0({\boldsymbol {\tilde
\rho}})}{\sum_{k \in \CN \setminus \{i\}} {\tilde \rho}_k}$, ${\tilde \rho}_i =
h_i^0({\boldsymbol {\tilde \rho}}) + \epsilon_i {\sum_{k \in \CN \setminus \{i\}} {\tilde \rho}_k}, i \in \CN$. Next, set $\epsilon =
\min_{i \in \CN} \epsilon_i$. Then ${\tilde \rho}_i
\geq h_i^0({\boldsymbol {\tilde \rho}}) + \epsilon {\sum_{k \in \CN \setminus \{i\}} {\tilde \rho}_k}, i \in \CN$. Thus for this value of $\epsilon$,
${\tilde {\boldsymbol \rho}} \in S_{{\boldsymbol \rho} \geq
{\boldsymbol h}^{\epsilon} ({\boldsymbol \rho})}$, and the result
follows.
\end{proof}

\begin{lemma}
\label{theo:epsiloninfinity}
Consider any ${\boldsymbol {\bar \rho}} > 0$ and any $\epsilon > 0$.
Denote by $\lambda$ a positive number.  For any $i \in \CN$,
$\displaystyle \lim_{\lambda \rightarrow
\infty} [h_i^\epsilon(\lambda {\boldsymbol {\bar \rho}})
- f_i(\lambda {\boldsymbol {\bar \rho}})] = \infty$.
\end{lemma}

\begin{proof}
Consider the definitions of $h_i^\epsilon(\lambda {\boldsymbol {\bar \rho}})$ and $f_i(\lambda {\boldsymbol {\bar \rho}})$. After some straightforward re-writing and
ignoring the constant term ${\boldsymbol f}(0)$ in
$h_i^\epsilon(\lambda {\boldsymbol {\bar \rho}})$,
the difference between the two functions has
the following form.
\begin{eqnarray}
&
\displaystyle \sum_{j \in \CJ_i} \frac{\ln(2)}{a_j}\left[
\displaystyle \sum_{k \in \CN \setminus \{i\}} b_{ikj}{\bar \rho}_k \lambda - \frac{1}{\ln\left(1 +
\frac{1}{\sum_{k \in \CN \setminus \{i\}}b_{ikj} {\bar \rho}_k \lambda + c_{ij}}\right)}\right] + \left( \displaystyle \sum_{k \in \CN \setminus \{i\}} \sum_{j \in \CJ_i}
\frac{\ln(2) b_{ikj}}{a_j} {\bar \rho}_k \right) \lambda &  \label{eq:diff}
\end{eqnarray}

Let $q(\lambda)$ denote the expression in the square brackets of
\eqref{eq:diff}. By repeatedly using l'H{\^ o}pital's rule, one
can show that $\lim_{\lambda \rightarrow \infty} q(\lambda) =
-\frac{1}{2} - c_{ij}$, which is a constant. Observing that the last term
in \eqref{eq:diff} grows linearly in $\lambda$, the lemma follows.
\end{proof}

\begin{theorem}
\label{theo:sufficient}
If $S_{{\boldsymbol \rho} = {\boldsymbol h}^0 ({\boldsymbol \rho})}
\not= \emptyset$, then $S_{{\boldsymbol \rho} = {\boldsymbol f}
({\boldsymbol \rho})} \not= \emptyset$.
\end{theorem}

\begin{proof}
By Lemma \ref{theo:epsilonexistence}, there exists $\epsilon > 0$ and
${\boldsymbol \rho}_{\boldsymbol h}^\epsilon$ satisfying ${\boldsymbol
\rho}_{\boldsymbol h}^\epsilon = {\boldsymbol h}^{\epsilon} ({\boldsymbol
\rho}_{\boldsymbol h}^\epsilon)$. It is easily verified that $\lambda
{\boldsymbol \rho}_{\boldsymbol h}^\epsilon \geq {\boldsymbol h}^{\epsilon}
(\lambda {\boldsymbol \rho}_{\boldsymbol h}^\epsilon), \lambda \geq 1$. Using Lemma
\ref{theo:epsiloninfinity}, there exists $\bar \lambda$ such that
${\bar \lambda} {\boldsymbol \rho}_{\boldsymbol h}^\epsilon \geq {\boldsymbol h}^{\epsilon}
({\bar \lambda} {\boldsymbol \rho}_{\boldsymbol h}^\epsilon) \geq
{\boldsymbol f}({\bar \lambda} {\boldsymbol \rho}_{\boldsymbol h}^\epsilon)$.
Therefore $S_{{\boldsymbol \rho} \geq {\boldsymbol f}
({\boldsymbol \rho})} \not= \emptyset$, and the result follows from
Corollary \ref{theo:opttoexistence}.
\end{proof}

Theorems \ref{theo:necessary} and \ref{theo:sufficient} together
provide a complete answer to the solution existence of LTE load
coupling, that is, whether or not the system has a fixed point
in $\mathbb{R}^n_+$ is equivalent to the feasibility of the linear equation system
${\boldsymbol \rho} = {\boldsymbol h}^0 ({\boldsymbol
\rho})$. Clearly, given an LTE network design, this feasibility check should be
performed first, before determining the load values. Furthermore, from
Theorem \ref{theo:necessary}, violating ${\boldsymbol
\rho}_{\boldsymbol h}^0\leq{1}$ is a simple indication of that
${\boldsymbol \rho}^*$ is beyond the network capacity.
For a two-cell example, the solution to the linear system ${\boldsymbol \rho} = {\boldsymbol h}^0 ({\boldsymbol
\rho})$ is
\begin{subequations}
\label{eq:twocellcapacity}
\begin{align}
& \rho_1=\frac{f_1({\bf 0})+f_2({\bf 0})\cdot H_{12}}{1-H_{21}H_{12}}, \\
& \rho_2=\frac{f_2({\bf 0})+f_1({\bf 0})\cdot H_{21}}{1-H_{21}H_{12}}.
\end{align}
\end{subequations}
%
%
%
%

With \eqref{eq:twocellcapacity}, a feasible solution exists when
$1-H_{21}H_{12}>0$, i.e., $H_{12}=\frac{1}{H_{21}}$ forms the (open)
boundary of the feasibility region in the two coefficients.
Note that $H_{12}$ and $H_{21}$ are linear in the traffic demands to
be satisfied in cell 1 and cell 2, respectively. The derived relation
representing the resource sharing trade-off for the two neighbor cells
in this example is well in line with the commonly known radio resource
sharing and capacity region concepts.

\section{Convex Optimization and Upper Bounding}
\label{sec:optimization}

Provided that $S_{{\boldsymbol \rho} = {\boldsymbol f} ({\boldsymbol
\rho})} \not= \emptyset$, a solution algorithm needs to be applied
to find ${\boldsymbol \rho}^*$. Solving ${\boldsymbol \rho} = {\boldsymbol
f} ({\boldsymbol \rho})$ is equivalent to finding the (unique) root of the $n$-dimensional function
${\boldsymbol \rho} - {\boldsymbol f} ({\boldsymbol
\rho})$. Thus one approach is to use the Newton-Raphson method.
In this section, we show that approaching ${\boldsymbol \rho}^*$ can
alternatively be viewed as the convex optimization problem formulated
below.
\begin{subequations}
\label{eq:maxrho}
\begin{align}
\max~~ & \sum_{i \in {\cal N}} \rho_i \label{eq:maxrhoobj}\\
& {\boldsymbol \rho} -  {\boldsymbol f}({\boldsymbol \rho}) \leq 0 \label{eq:maxrhocons} \\
& {\boldsymbol \rho} \in  \mathbb{R}_+^n
\end{align}
\end{subequations}

\begin{corollary}
\label{theo:maxrho}
Formulation \eqref{eq:maxrho} is a convex optimization problem,
and if $S_{{\boldsymbol \rho} = {\boldsymbol f} ({\boldsymbol \rho})}
\not= \emptyset$, then ${\boldsymbol \rho}^*$ is
the unique optimum to \eqref{eq:maxrho}.
\end{corollary}

\begin{proof}
Because ${\boldsymbol f}({\boldsymbol \rho})$ is concave (Theorem
\ref{theo:concave}), ${\boldsymbol \rho}-{\boldsymbol f}({\boldsymbol
\rho})$ is convex in ${\mathbb R}^n_+$.
Thus ${\boldsymbol \rho}-{\boldsymbol f}({\boldsymbol
\rho}) \leq 0$ is a convex set. The proof is complete by observing that,
similar to \eqref{eq:optrho}, optimum to \eqref{eq:maxrho} must satisfy
\eqref{eq:maxrhocons} with equality, and
${\boldsymbol \rho}^*$ is the unique solution to ${\boldsymbol \rho} =
{\boldsymbol f} ({\boldsymbol \rho})$.
\end{proof}

Following Corollary \ref{theo:maxrho}, any convex optimization solver
can be used to approach ${\boldsymbol \rho}^*$.  In network planning,
one will need to solve
\eqref{eq:rho} repeatedly to evaluate many candidate BS location and antenna
configurations. Typically, the performance evaluation does not have to
be exact in order to relate the quality of a candidate solution to
that of another. Utilizing the structure of ${\boldsymbol f}$, we can
numerically obtain upper bounds to ${\boldsymbol \rho}^*$ via linear
equations. Consider any ${\boldsymbol {\bar \rho}} \in {\mathbb
R}_+^n$.  Using the partial derivatives of $\boldsymbol f$ at
${\boldsymbol {\bar \rho}}$, and the point ${\boldsymbol
f}({\boldsymbol {\bar \rho}})$, we obtain an upper approximation of
$\boldsymbol f$ due to concavity. Formally, denote by ${\boldsymbol {\bar
h}}$ the linear function of \eqref{eq:linear} where ${\boldsymbol
{\hat \rho}} = {\boldsymbol {\bar \rho}}$ and $H_{ik}$, defined by
\eqref{eq:partial}, takes the following value,
$$
H_{ik} = \frac{\partial f_i}{\partial \rho_k}({\boldsymbol {\bar \rho}})
= \sum_{j \in \CJ_i} \ln(2) \frac{b_{ikj}}{a_j} \frac{1}{\ln^2\left(1+\frac{1}{
\sum_{h \in \CN \setminus \{i\}}  b_{ihj} {\bar \rho}_h + c_{ij}
}\right) \left(\sum_{h \in \CN \setminus \{i\}}  b_{ihj} {\bar \rho}_h + c_{ij}\right)^2
\left(1 + \frac{1}{
\sum_{h \in \CN \setminus \{i\}}  b_{ihj} {\bar \rho}_h + c_{ij}
}\right)}.
$$

The positive-valued solution to the linear system ${\boldsymbol \rho}
= {\boldsymbol {\bar h}}({\boldsymbol \rho})$, if it exists, is denoted
by ${\boldsymbol \rho}_{\boldsymbol {\bar h}}$. As established below,
${\boldsymbol {\bar h}}$ and ${\boldsymbol \rho}_{\boldsymbol {\bar h}}$
yield upper estimations of $\boldsymbol f$ and ${\boldsymbol \rho}^*$,
respectively.

\begin{corollary}
\label{theo:above}
${\boldsymbol {\bar h}}({\boldsymbol \rho}) \geq {\boldsymbol f}({\boldsymbol \rho}), {\boldsymbol \rho} \geq 0$.
\end{corollary}

\begin{proof}
Follows immediately from the concavity of ${\boldsymbol f}$
and the definition of ${\boldsymbol {\bar h}}$.
\end{proof}

\begin{theorem}
\label{theo:upperbound}
If $S_{{\boldsymbol \rho} = {\boldsymbol {\bar h}}({\boldsymbol \rho})} \not= \emptyset$,
then ${\boldsymbol \rho}_{\boldsymbol {\bar h}} \geq {\boldsymbol \rho}^*$.
\end{theorem}

\begin{proof}
Consider the linear programming (LP) formulation $\max \{\sum_{i \in
\CN} \rho_i : {\boldsymbol \rho} \leq {\boldsymbol {\bar
h}}({\boldsymbol \rho}), {\boldsymbol \rho} \in {\mathbb
R}_+^n\}$. Similar to \eqref{eq:maxrho}, it is easily realized that
the LP formulation, if feasible, has a unique optimum satisfying
${\boldsymbol \rho} = {\boldsymbol {\bar h}}({\boldsymbol
\rho})$. Hence the unique optimum is ${\boldsymbol \rho}_{\boldsymbol
{\bar h}}$. Moreover, ${\boldsymbol \rho}^*$ is, by Corollary
\ref{theo:above}, a feasible point to the LP, and hence the LP
remains feasible after including ${\boldsymbol \rho} \geq {\boldsymbol
\rho}^*$, and the result follows.
\end{proof}

The process of solving the load-coupling system, e.g., an interior
point method for \eqref{eq:maxrho}, will typically generate a sequence
of iterations approaching ${\boldsymbol \rho}^*$ from below. By Theorem
\ref{theo:upperbound}, the iterations can be used to compute upper bounds,
thus yielding an interval confining ${\boldsymbol \rho}^*$.
In order to speed up the process of network optimization,
performance evaluation of a candidate planning solution can
use a threshold of the maximum size of the interval, instead of
computing the exact solution of the load vector.

Computing an upper bound ${\boldsymbol
\rho}_{\boldsymbol {\bar h}}$, involves solving a system of $n$ linear equations.
The same amount of computation applies to the feasibility check and
computing lower bounding ${\boldsymbol \rho}_{\boldsymbol h}^0$ in Section~\ref{sec:linear}.
It is straightforward to see that calculating the coefficients is of
complexity $\mathcal{O}(n^2)$. Thus the overall complexity lies in the
matrix inversion operation that runs in the time range
$\mathcal{O}(n^{2.3727})$ and $\mathcal{O}(n^3)$, where the former is
attainable only asymptotically by the Coppersmith--Winograd algorithm.

\section{Numerical Results}
\label{sec:experiment}

In this section, we numerically investigate the theoretical findings
in the previous sections. An illustrative simulation study has been
conducted for a three-site 3GPP LTE network with an inter-site
distance of 500~m, adopting a wrap-around technique. The simulated
system operates at 2~GHz with 10~MHz bandwidth. Each site is equipped
with a three-sector downtilted directional antenna with 14~dBi antenna
gain. The propagation environment and user distribution follow the
3GPP specification in \cite{3gpp36814}, assuming propagation model 1
(Okumura-Hata, urban, 8~dB standard deviation shadow fading) and user
generation scenario 4b with one hotspot of 40~m radius per macro cell
area.  Note that, as for any system model, the complete assessment of
the model validity would also include validating of numerical results
against results from real deployments, which is beyond the scope of
the current paper.

The network layout we have used is illustrated in
Figure~\ref{fig:experiment_configs}. Two layers of users are
generated, with 30 users per macro cell area in total, out of which
2/3 (the dot markers) is in a randomly placed hotspot, and 1/3 (the
x-markers) are distributed randomly and uniformly over the area. Each
user equipment has an omni-directional antenna with 0~dBi antenna
gain. The traffic demand corresponds to 400~kbit/s for all users
within a duration of one second in the time domain.


Two network configurations are illustrated in
Figure~\ref{fig:experiment_configs}, with the only difference being the
antenna direction of cell~1, which impacts the sets of users served by
the cell and its neighbors. Intuitively, configuration two
is inferior, since it results in that the hotspot users in cell~1's
original coverage area (see Figure~\ref{fig:experiment_config1}) is to
be served by cell~8 and/or cell~9, although these users are relatively
far away from the two cells.
The likely impact is poorer link
quality for the users in the handed-over hotspot as well as increased
number of users to be served by the neighbors of cell~1. These effects
are expected to be seen in the load of the neighbor cells.

First, for both configurations, the existence of system solutions has
been verified by finding $\boldsymbol\rho^0_{h}$ to the corresponding
linear system, as described in Section~\ref{sec:linear}. Next, the
non-linear coupling system \eqref{eq:rho} is solved using the
non-linear optimization toolbox of MATLAB.  Both
$\boldsymbol\rho^0_{h}$ and $\boldsymbol \rho^*$ are shown in
Figure~\ref{fig:experiment_load_hists}. As expected, the load of cells~8
and 9 increases for the second configuration. At the same time, the
load of cell~1 does not decrease either, even though it serves fewer
users under the second configuration. This is due to a joint effect of
several factors. Firstly, as can be seen from
Figure~\ref{fig:experiment_config2}, users served by cell~1 are likely
to experience high interference from cell~6 and vice versa.  From
Figure~\ref{fig:experiment_load_hist2}, we observe that the load of
cell~6 has also slightly increased. Secondly, the increased load in
cells~8 and 9 implies more frequent transmissions in these cells and
thus higher probability of interference to other cells; this in turn
increases the load of the other cells, which can also be clearly seen
in Figure~\ref{fig:experiment_load_hists}. We further note that the
solution $\boldsymbol \rho^*$ and thus the second configuration are
not feasible from the capacity point of view.


In Figure~\ref{fig:experiment_load_hists}, we also illustrate the load
solutions to the linear systems described in
Section~\ref{sec:optimization}, i.e., the upper bound
$\boldsymbol\rho_{{\boldsymbol {\bar h}}}$ assuming
$\boldsymbol{\widehat \rho} =
\boldsymbol \rho^{0}_{{\boldsymbol h}}$. Table~\ref{tab:bounds} provides further details
of the quality of both the lower and upper bounds obtained for all 9
cells in Configuration~1 for which the solution is illustrated in
Figure~\ref{fig:experiment_config1}.

The tight upper bound indicates the efficiency of the linear
approximation described in Section~\ref{sec:optimization}.  In
average, the estimation is deviates only a few percent from the true
load value. For $\boldsymbol{\rho^0_h}$ the values are significantly
lower than ${\boldsymbol \rho}^*$, as ${\boldsymbol \rho} =
{\boldsymbol h}^0({\boldsymbol \rho})$ represents a very optimistic
view of load coupling.  The observation sheds further light on the
importance of fundamental characterization that the two systems are
completely equivalent in solution existence, despite the large
difference in numerical values of $\boldsymbol{\rho^0_h}$ and
${\boldsymbol \rho}^*$. Improving the lower bounds,
although being beyond the scope of the current paper, is an
interesting topic for future investigation. For example, in
\cite{SiYu12}, the model discussed in this paper is applied for load
balancing, for which very tight lower and upper bounds are obtained
using few fixed-point iterations. It should be further noted that
although the results have been presented for downlink, the model and
the theoretical findings can be also applied to uplink.


For the two network configurations, in Figure~\ref{fig:feasibility} we
illustrate the behavior of the cell-load coupling with respect to
demand, which is successively scaled up uniformly over the service
area.  Figure~\ref{fig:feasibility_nonlinear} and
Figure~\ref{fig:feasibility_linear} show, respectively, the results for
the nonlinear load-coupling system
\eqref{eq:rho}, and the linear equation system ${\boldsymbol \rho} =
{\boldsymbol h}^0({\boldsymbol \rho})$ that provides lower estimation
and characterizes feasibility of \eqref{eq:rho}. The two
configurations are distinguished by using respectively solid and
dotted curves for the load solutions of nine cells. In Figure~\ref{fig:feasibility_nonlinear}, the thicker
curves represent the load of cell 8. For the linear system, only the
maximum value among the cells is shown in
Figure~\ref{fig:feasibility_linear} for the sake of clarity.

From Figure~\ref{fig:feasibility_nonlinear}, it is apparent that the
solution values of \eqref{eq:rho} grow rapidly in the high-demand
region, and becomes infeasible beyond some point. The feasibility
boundaries for the two configurations are shown by the vertical
lines. Configuration one is clearly superior, as its load values,
shown by the solid curves, are below those of configuration two, and
the feasibility boundary is considerably higher. For configuration
two, cell 8 has the highest load (the dotted thick line). Using
configuration one, some of the users are served by cell 1 instead (see
Figure~\ref{fig:experiment_configs}), leading to lower load in cell 8
(the solid thick line).  Note that, for both configurations, when
getting somewhat close to the infeasible region, the solver gives
solutions containing some zero elements, indicating the solution is
invalid (the system becomes unstable), before all the values abruptly drop to zeros showing
infeasibility.  When this occurs, the distance to the feasibility
boundary is, in fact, significant. The behavior shows the importance
of our analysis of characterizing feasibility exactly by the linear
system ${\boldsymbol
\rho} = {\boldsymbol h}^0({\boldsymbol \rho})$.

In Figure~\ref{fig:feasibility_linear}, the linear system gives
values growing consistently in demand.  The system gives the feasibility
boundary point of the non-linear load coupling equations, when the
determinant of ${\boldsymbol I}-{\boldsymbol H}$, where $\boldsymbol
I$ is the identity matrix and $\boldsymbol H$ is the matrix defined
for ${\boldsymbol h}^0$, equals zero. After passing the boundary
point, the linear system returns negative (infeasible) solutions. Hence
the numerical results verify that the significance of the linear system
to identifying solution existence. Moreover, using the linear system, one is able
to conclude, as shown in Figure~\ref{fig:feasibility}, that
configuration one is clearly superior to configuration two.

\section{Conclusions}
\label{sec:conclusion}

We have provided a theoretical analysis of the LTE load coupling
system originally presented in \cite{SiFuFo09} and have derived its
fundamental properties, including concavity, behavior in limit, and
solution uniqueness. We have also formulated the necessary and
sufficient condition for solution existence, The analysis leads to a
simple means for determining feasibility.  In addition, we have
presented two linear approximations. The analysis has been supported
by theoretical proofs and numerical experiments and can serve as a
fundamental basis for developing radio network planning and
optimization strategies for LTE. Furthermore, the presented
linearizations and the bounding-based optimization can potentially be
used for more general convex optimization problems with similar
properties.

\section*{Acknowledgments}

The work of the second author is supported by the Swedish Foundation
of Strategic Research (SSF).  The authors would also like to thank the
reviewers for their comments and suggestions.

\newpage
\begin{figure}[ht!]
\centering
\subfigure[Feasible solution within network capacity: $0\leq{\boldsymbol \rho}\leq{1}$.]{
\includegraphics[scale=0.22]{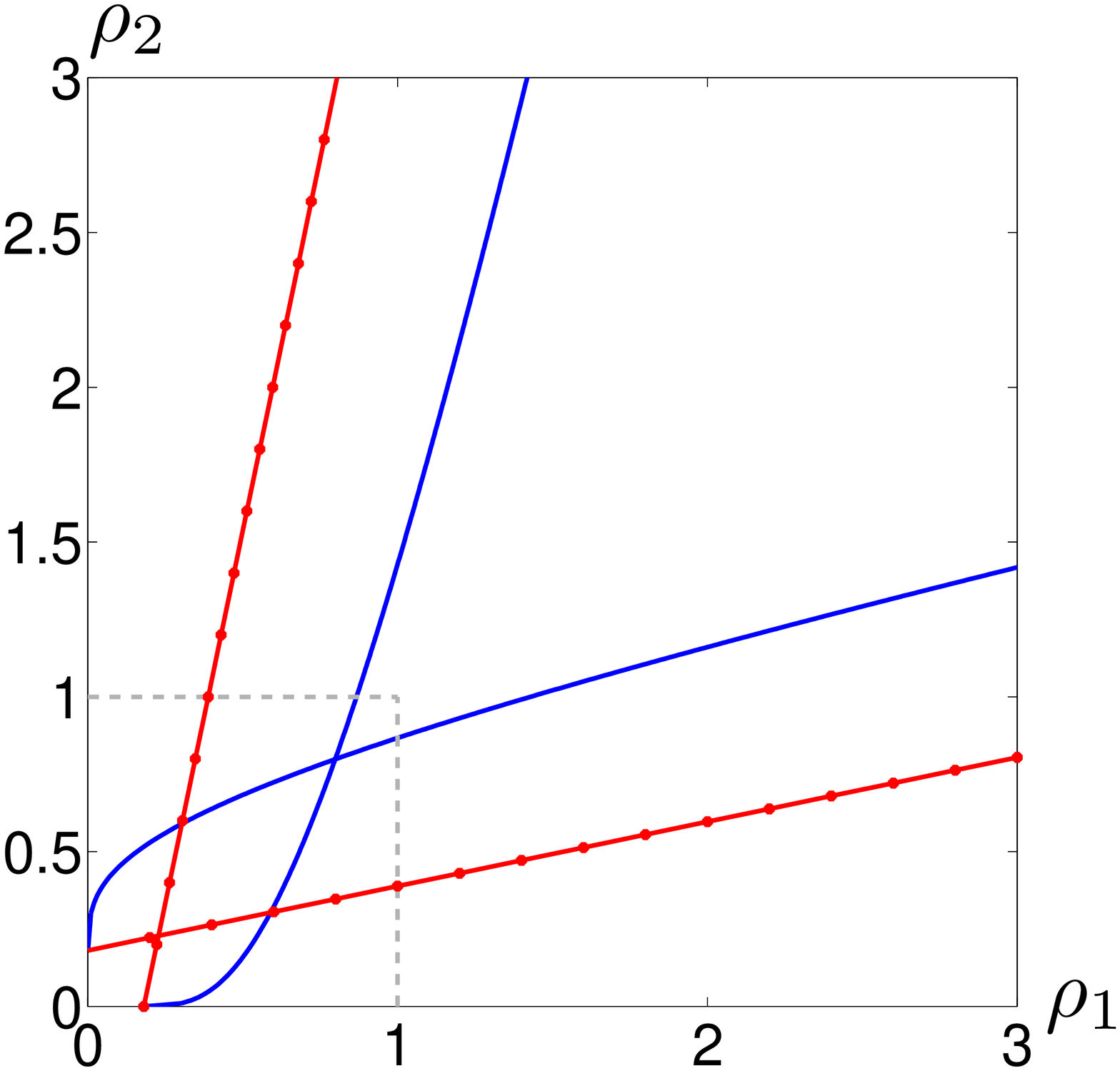}
\label{fig:case1}
}
\subfigure[Feasible solution beyond network capacity: $\exists{(i\in{\mathcal{N}})}|\rho_{i}>1$.]{
\includegraphics[scale=0.22]{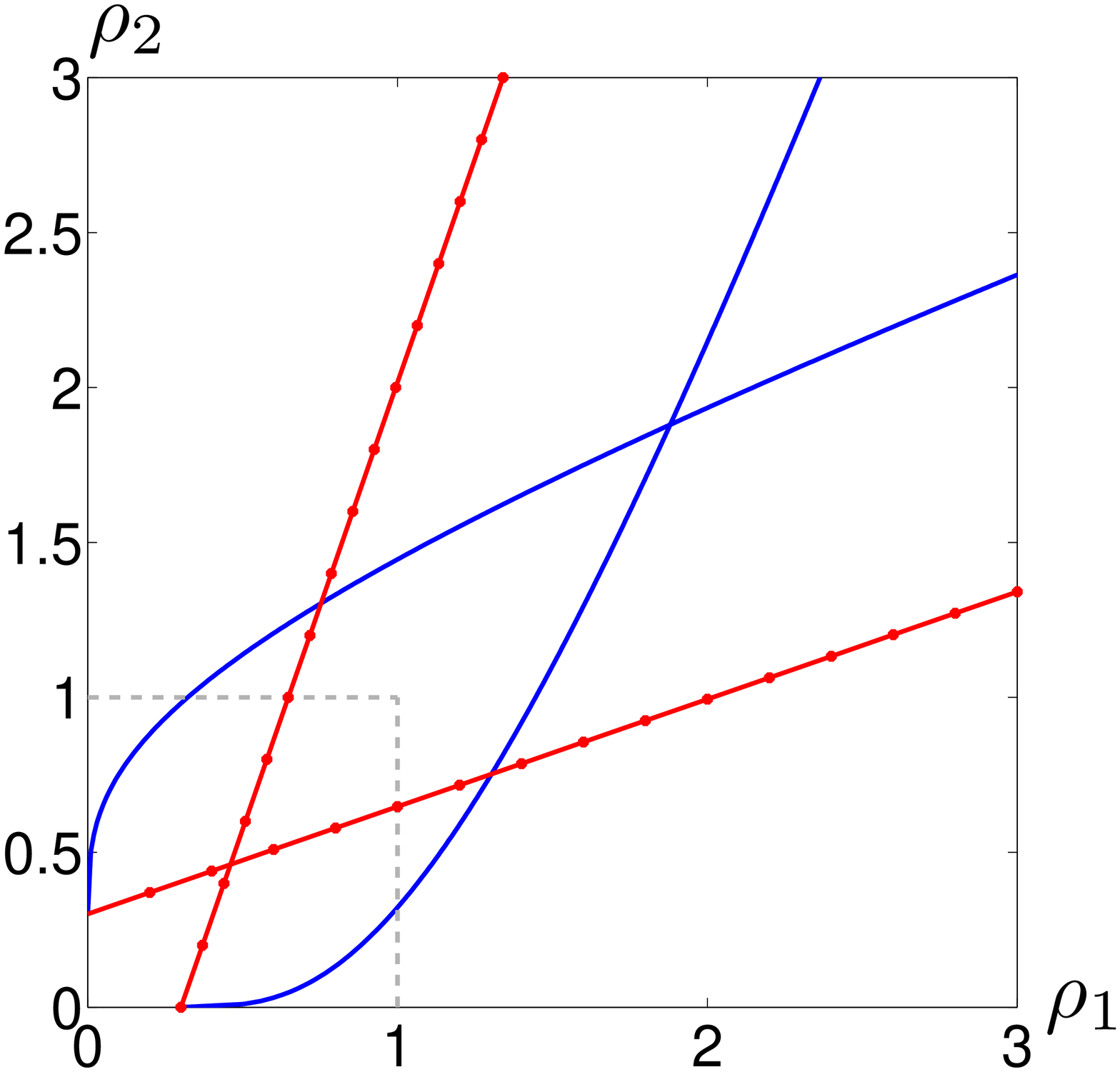}
\label{fig:case2}
}
\subfigure[Infeasible system: $S_{{\boldsymbol \rho} =  {\boldsymbol f}({\boldsymbol \rho})} = \emptyset$]{
\includegraphics[scale=0.22]{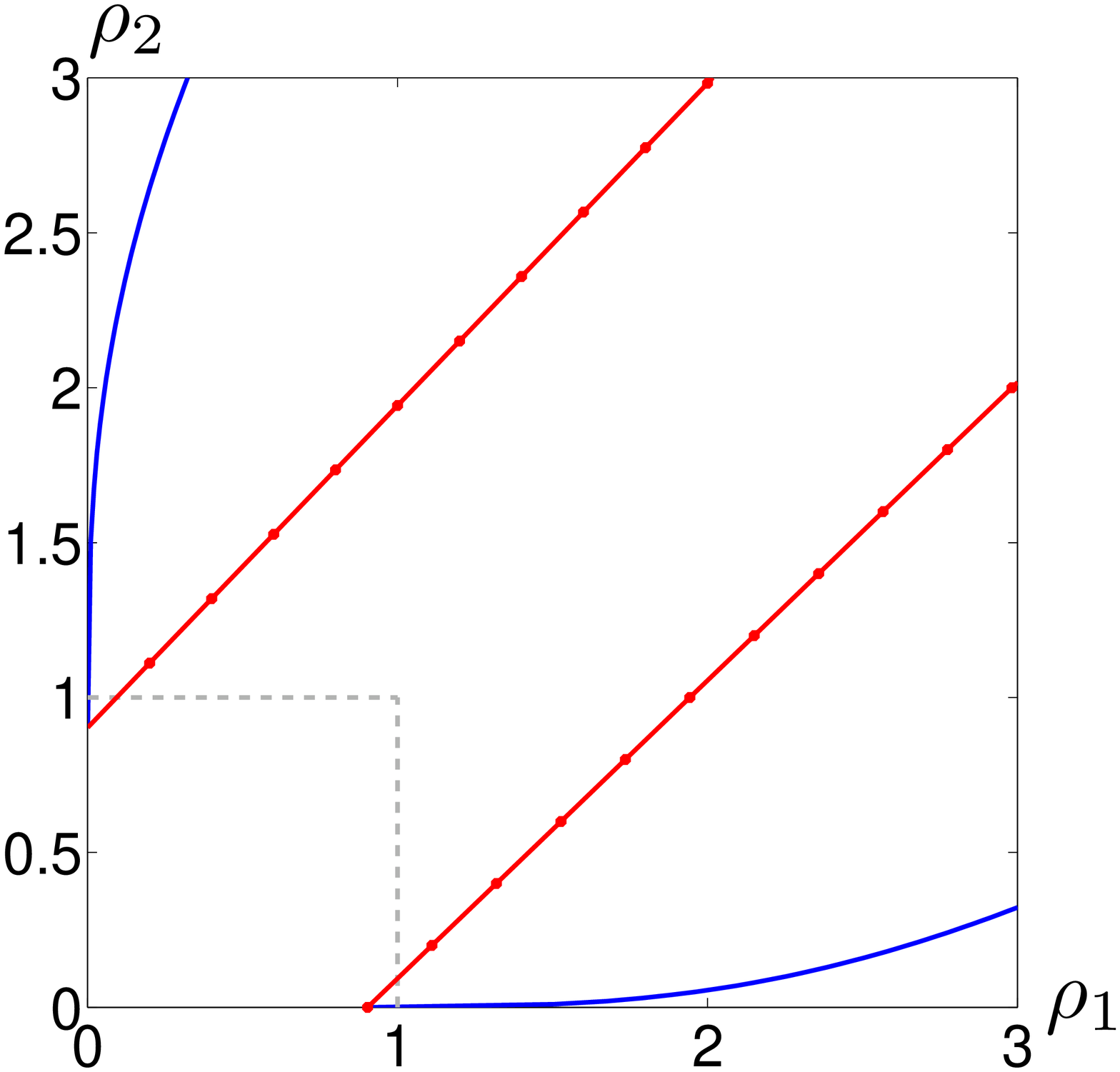}
\label{fig:case3}
}
\caption[]{An illustration of the load-coupling system of two cells.}
\label{fig:twocell}
\end{figure}

\begin{figure}[ht!]
\centering
\subfigure[Configuration 1.]{
\includegraphics[scale=0.44]{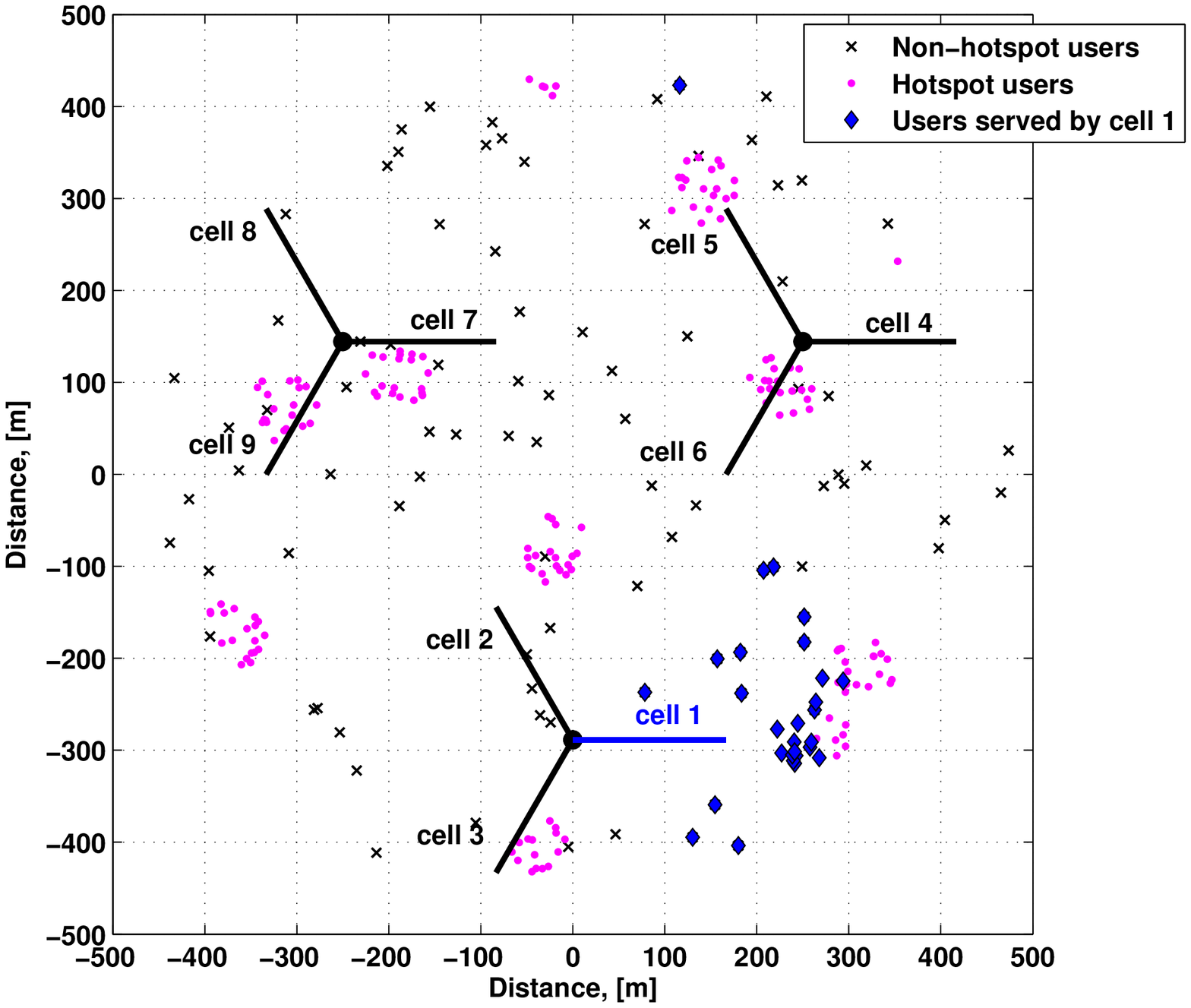}
\label{fig:experiment_config1}
}
\subfigure[Configuration 2.]{
\includegraphics[scale=0.44]{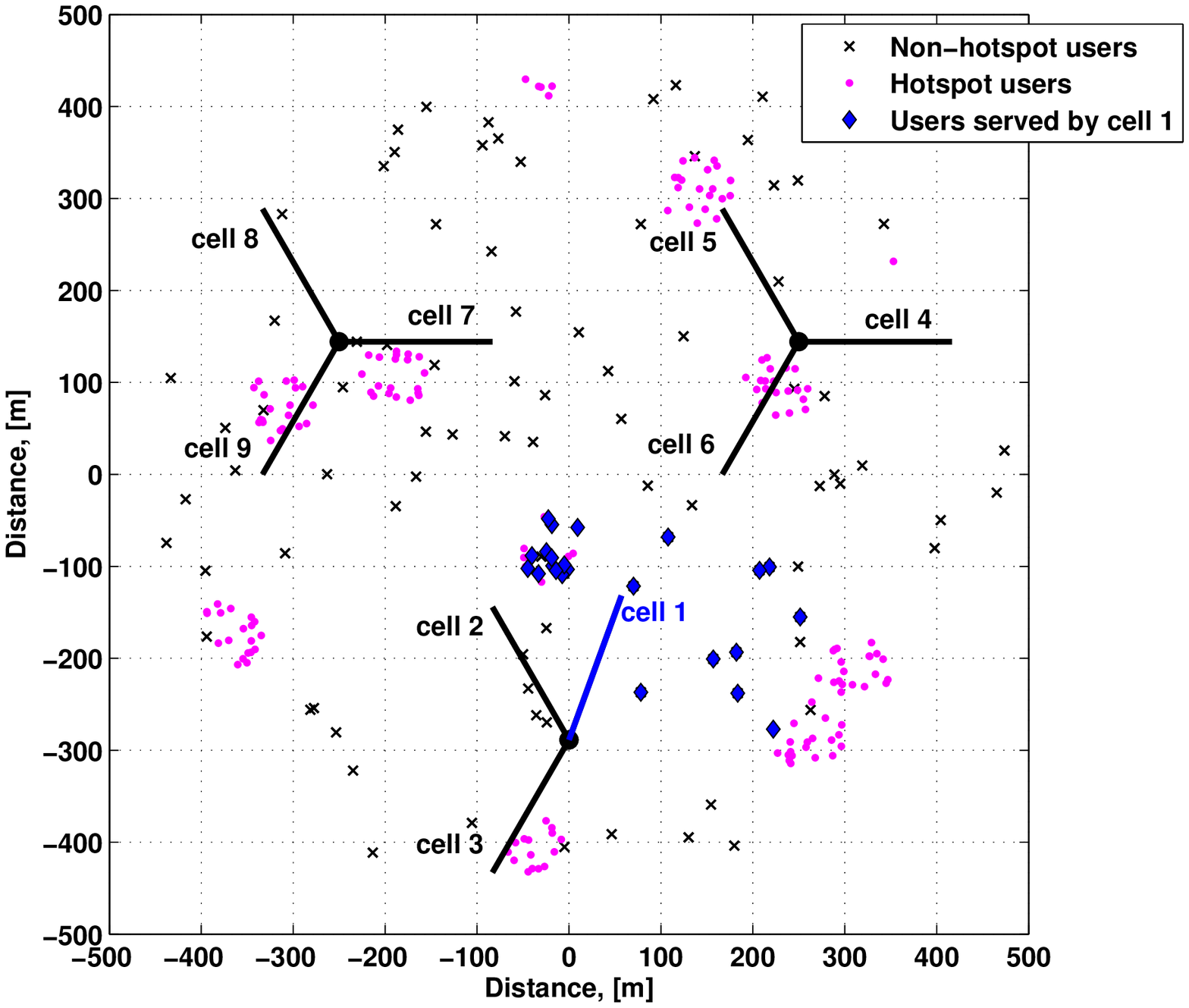}
\label{fig:experiment_config2}
}
\caption[]{Network configurations for numerical studies.}
\label{fig:experiment_configs}
\end{figure}

\begin{figure}[ht!]
\centering
\subfigure[Configuration 1: Feasible solution within network capacity.]{
\includegraphics[scale=0.43]{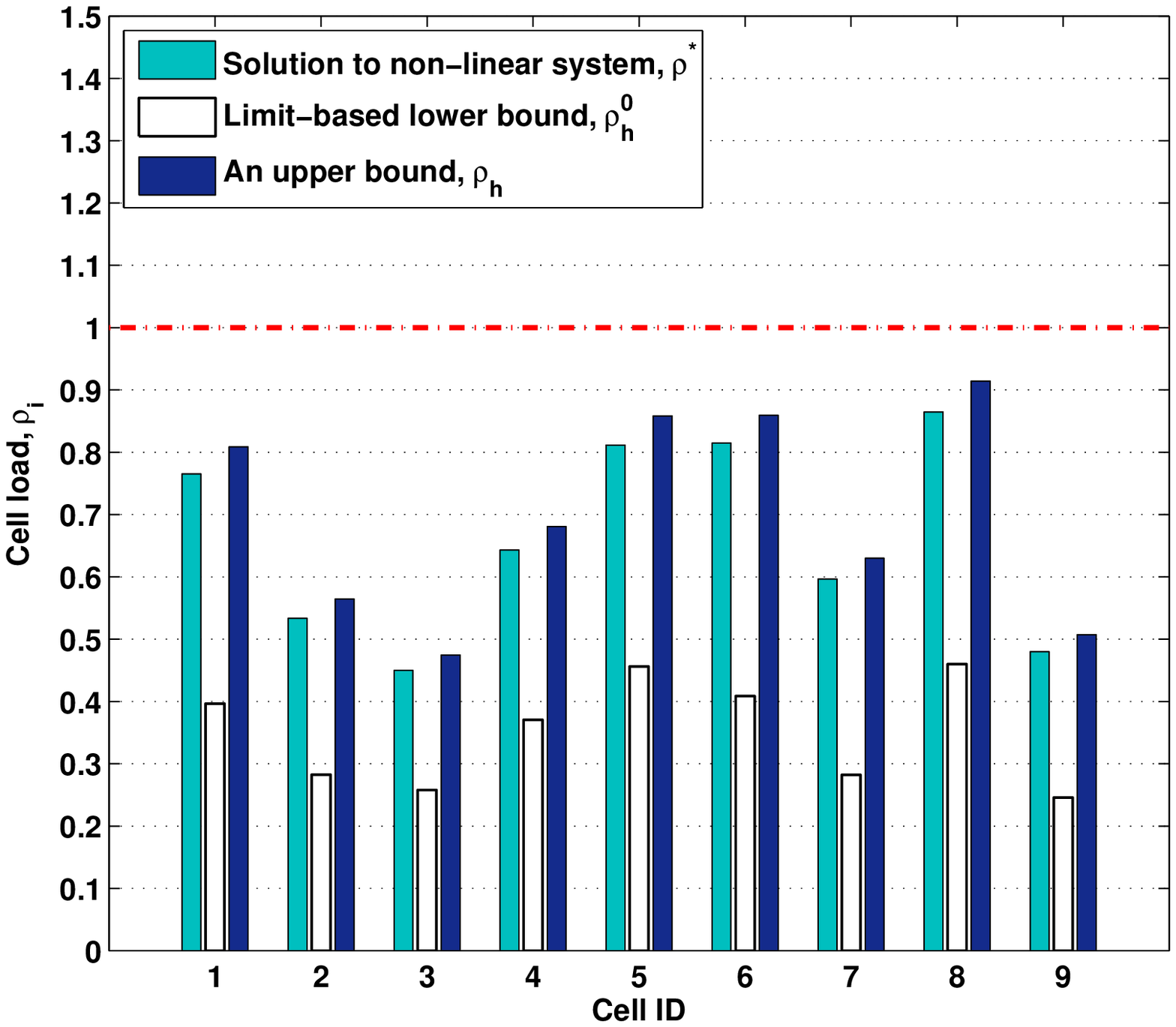}
\label{fig:experiment_load_hist1}
}
\subfigure[Configuration 2: Feasible solution beyond network capacity.]{
\includegraphics[scale=0.43]{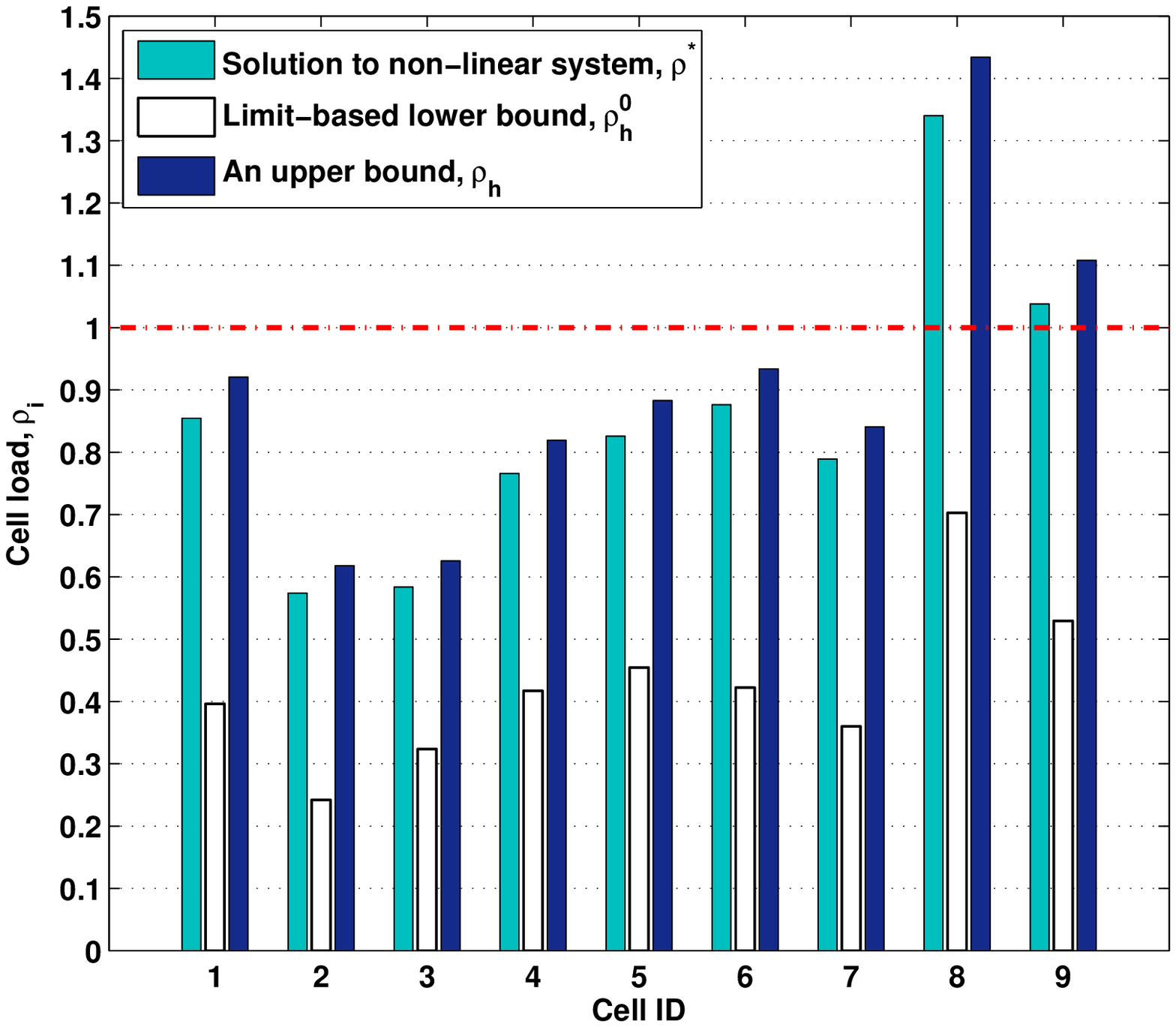}
\label{fig:experiment_load_hist2}
}
\caption[]{Load solutions for the two example network configurations.}
\label{fig:experiment_load_hists}
\end{figure}

\begin{table}[!ht]
\small
\caption{Lower and upper bounds quality for Configuration 1}
\label{tab:bounds}
\centering
\begin{tabular}{|c|c|c|c|c|c|c|c|c|c|}
\hline
Bound & Cell 1 & Cell 2 & Cell 3 & Cell 4 & Cell 5 & Cell 6 & Cell 7 & Cell 8 & Cell 9 \\
\hline
Upper bound, $\frac{|\boldsymbol{\rho^0_h}-\boldsymbol{\rho^*}|}{\boldsymbol{\rho^*}}\cdot{100\%}$ & 5.19 & 5.55 & 6.66 & 4.61 & 4.93 & 4.94 & 3.33 & 5.20 & 4.17\\
Lower bound, $\frac{|\boldsymbol{\rho_{\bar{h}}}-\boldsymbol{\rho^*}|}{\boldsymbol{\rho^*}}\cdot{100\%}$ & 48.05 & 48.15 & 44.44 & 43.08 & 43.21 & 50.62 & 51.66 & 46.82 & 50.00 \\
\hline
\end{tabular}
\end{table}

\begin{figure}[ht!]
\centering
\subfigure[Growth of cell load with respect to demand for the load coupling system \eqref{eq:rho}.]{
\includegraphics[scale=0.4]{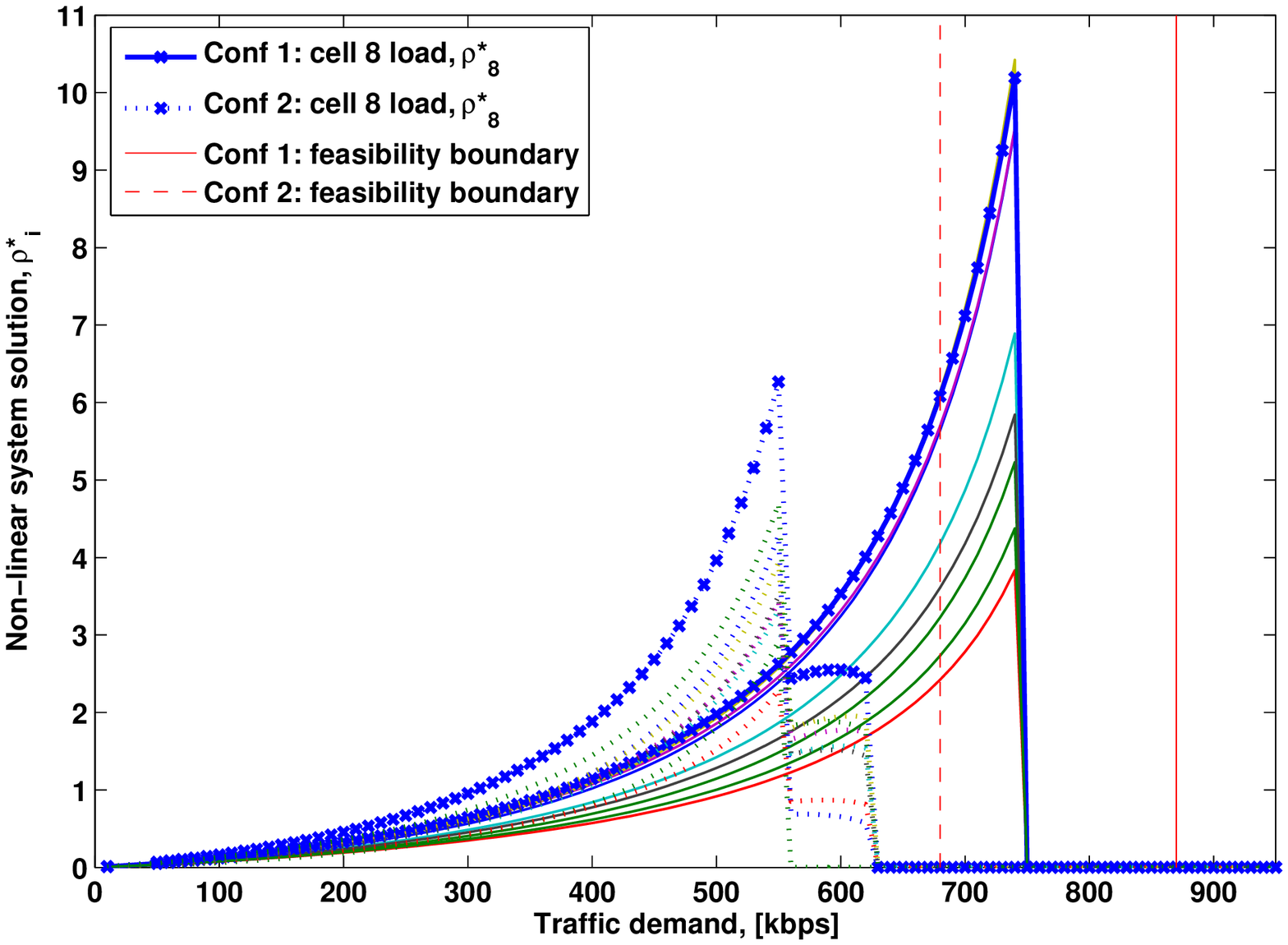}
\label{fig:feasibility_nonlinear}
}~~
\subfigure[Growth of cell load with respect to demand for the linear system
${\boldsymbol \rho} = {\boldsymbol h}^0({\boldsymbol \rho})$.]{
\includegraphics[scale=0.4]{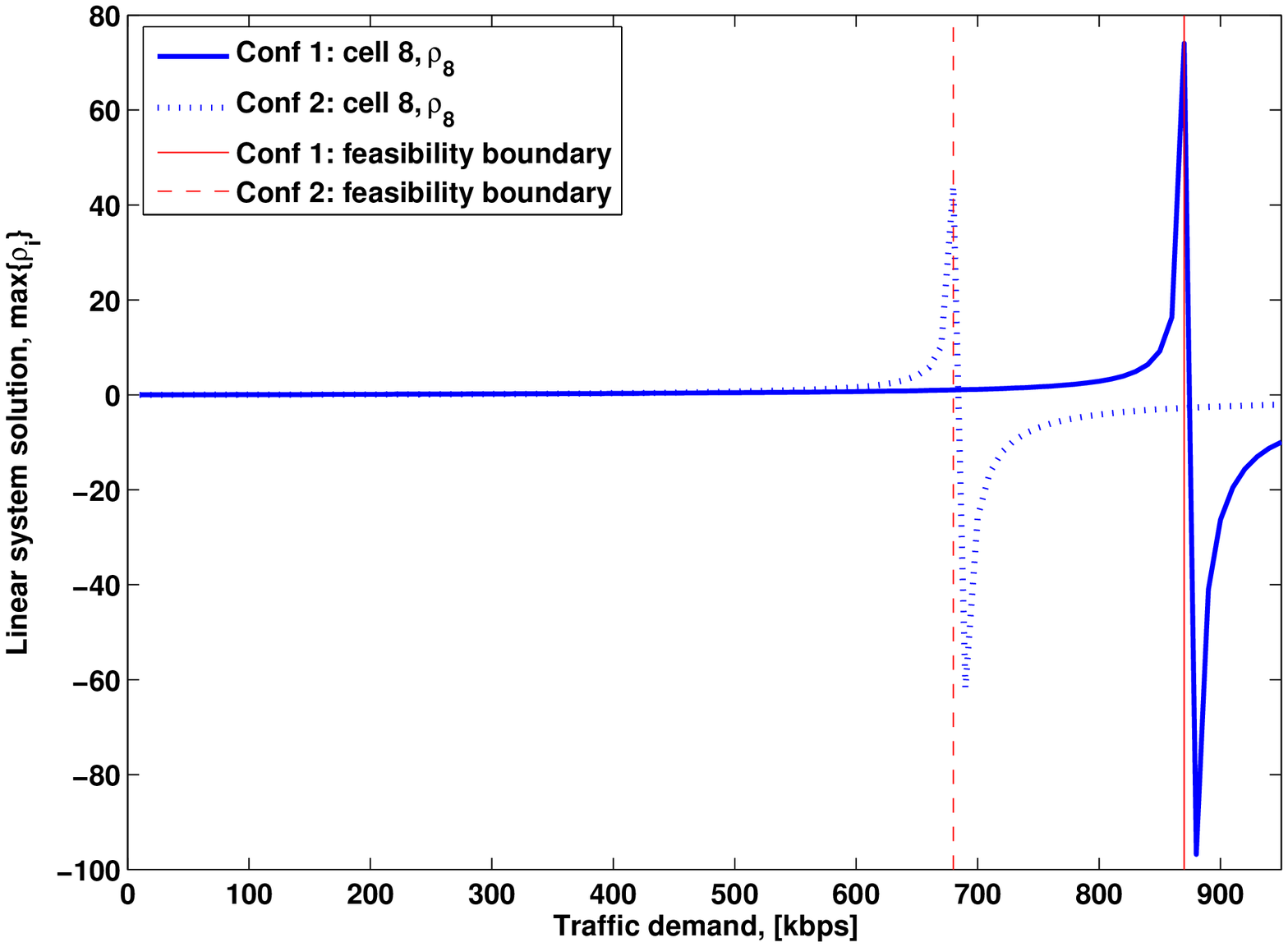}
\label{fig:feasibility_linear}
}
\caption[]{System solution with respect to traffic demand.}
\label{fig:feasibility}
\end{figure}

%









\begin{thebibliography}{99}

\bibitem{3gpp36814}
3GPP TS 36.814.
\newblock Evolved Universal Terrestrial Radio Access (E-UTRA); Further advancements for E-UTRA physical layer aspects, v.9.0.0.
\newblock \url{http://www.3gpp.org}

\bibitem{Aein73}
J.~M.~Aein.
\newblock Power balancing in systems employing frequency reuse.
\newblock{\em COMSAT Technical Review}, 3(2):277--300, 1973.

\bibitem{AmCaMa03}
E.~Amaldi, A.~Capone, and F.~Malucelli.
\newblock Planning UMTS base station location: optimization models with power control and
algorithms.
\newblock {\em IEEE Transactions on Wireless Communications}, 2:939--952, 2003.

\bibitem{AmCaMa08}
E.~Amaldi, A.~Capone, and F.~Malucelli.
\newblock Radio planning and coverage optimization of 3G cellular networks.
\newblock {\em Wireless Networks}, 14:435--447, 2008.


\bibitem{AmCaMa06}
E.~Amaldi, A.~Capone, F.~Malucelli, and C.~Mannino.
\newblock Optimization problems and models for planning cellular networks.
In M.~Resende and P.~Pardalos, editors,
\newblock{\em Handbook of Optimization in Telecommunications}. Springer Science, 2006.

\bibitem{AsMo08}
M.~Assaad and A.~Mourad.
\newblock New frequency-time scheduling algorithms for 3GPP/LTE-like OFDMA air interface in the
downlink.
\newblock {\em Proc.\ of IEEE VTC Spring '08}, 2008.

\bibitem{BoNaAl11}
H.~Boche, S.~Naik, and T.~Alpcan.
\newblock Characterization of convex and concave resource allocation problems in interference coupled
wireless systems.
\newblock {\em IEEE Transactions on Signal Processing}, 59:2382--2394, May.~2011.

\bibitem{BoNaSc11a}
H.~Boche, S.~Naik, and M~Schubert.
\newblock Combinatorial characterization of interference coupling in wireless systems.
\newblock {\em IEEE Transactions on Signal Processing}, 59:1697--1706, Jun.~2011.

\bibitem{BoNaSc11b}
H.~Boche, S.~Naik, and M~Schubert.
\newblock Pareto boundary of utility sets for multiuser wireless systems.
\newblock {\em IEEE Transactions on Networking}, 19:589--602, Apr.~2011.

\bibitem{BoSc07}
H.~Boche and M.~Schubert.
\newblock Multiuser interference balancing for general interference functions -- a convergence analysis.
\newblock {\em Proc.\ of IEEE ICC '07}, 2007.

\bibitem{BoSt08}
H.~Boche and S.~Sta{\'n}czak.
\newblock Strict convexity of the feasible log-SIR region.
\newblock {\em IEEE Transactions on Communications}, 56:1511--1518, Sep.~2008.

\bibitem{CaImMa04}
D.~Catrein, L.~A.~Imhof, and R.~Mathar.
\newblock Power control, capacity, and duality of uplink and downlink in cellular CDMA systems.
\newblock {\em IEEE Transactions on Communications}, 52:1777--1785, Oct. 2004.

\bibitem{DaPaSkBe07}
E.~Dahlman, S.~Parkvall, J.~Sk{\"o}ld
\newblock {\em 4G: LTE/LTE-Advanced for Mobile Broadband}. Elsevier Science \& Technology, 2011.

\bibitem{EiGe08}
A.~Eisenbl{\"a}tter and H.-F.~Geerdes.
\newblock Capacity optimization for UMTS: bounds and benchmarks for interference reduction.
\newblock {\em Proc.\ of IEEE PIMRC '08}, 2008.

\bibitem{EiGeKoMaWe05}
A.~Eisenbl{\"a}tter, H.-F.~Geerdes, T.~Koch,
A.~Martin, and R.~Wess{\"a}ly.
\newblock UMTS radio network evaluation and optimization beyond snapshots.
\newblock {\em Mathematical Methods of Operations
Research}, 63:1--29, 2005.

\bibitem{EiKoMaAcFuKoWeWe02}
A.~Eisenbl{\"a}tter, T.~Koch, A.~Martin, T.~Achterberg, A.~F{\"u}genschuh, A.~Koster, O.~Wegel, and
R.~Wess{\"a}ly.
\newblock Modelling feasible network configurations for UMTS.
\newblock In: G.~Anandalingam, and S.~Raghavan, editors,
{\em Telecommunications Network Design and Management}. Kluwer
Academic Publishers, 2002.

\bibitem{Ek09}
H.~Ekstr{\"o}m.
\newblock QoS control in the 3GPP evolved packet system.
\newblock {\em IEEE Communications Magazine}, 76--83, February 2009.


\bibitem{GaRuOl04}
M.~Garcia-Lozano, S.~Ruiz, and J.~J.~Olmos.
\newblock UMTS optimum cell load balancing for inhomogeneous traffic patterns.
\newblock {\em Proc.\ of IEEE VTC Fall '04}, 2004.

\bibitem{GrViGo92}
S.~A.~Grandhi, R.~Vijayan, and D.~J.~Goodman.
\newblock Distributed algorithm for power control in cellular radio systems.
\newblock{\em Proc.\ of 30th Allerton Conf.\ on Communication, Control, and Computing}, Sep. 1992.


\bibitem{JaLe03}
J.~Jan and K.~B.~Lee.
Transmit power adaptation for multiuser OFDM systems.
{\em IEEE Journal on Selected Areas in Communications}, 21:171--178, 2003.

\bibitem{LePeMeXuLu09}
S.-B.~Lee, I.~Pefkianakis, A.~Meyerson, S.~Xu, and S.~Lu.
\newblock Proportional fair frequency-domain packet scheduling for 3GPP LTE uplink.
\newblock {\em Proc.\ of IEEE INFOCOM '09}, pp. 2611--2615, 2009.

\bibitem{LeFaZhYa07}
H.~Lei, C.~Fan, X.~Zhang, and D.~Yang.
\newblock QoS aware packet scheduling algorithm for OFDMA systems.
\newblock {\em Proc.\ of IEEE VTC Fall '07}, 2007.

\bibitem{MaKo10}
K.~Majewski and M.~Koonert.
\newblock Conservative cell load approximation for radio networks with Shannon channels
and its application to LTE network planning.
\newblock {\em Proc.\ of the 5th Advanced International Conference on Telecommunications}, 2010.

\bibitem{MaTuHuBo07}
K.~Majewski, U.~T{\"u}rke, X.~Huang, and B.~Bonk.
\newblock Analytical cell load assessment in OFDM radio networks.
\newblock {\em Proc.\ of IEEE PIMRC '07}, 2007.

\bibitem{MaSc01}
R.~Mathar and M.~Schmeink.
\newblock Optimal base station positioning and channel assignment for {3G} mobile
networks by integer programming.
\newblock {\em Annals of Operations Research}, 107:225--236, 2001.

\bibitem{MeHe01}
L.~Mendo and J.~M.~Hernando.
\newblock On dimension reduction for the power control problem.
\newblock {\em IEEE Transactions on Communications}, 49(2):243--248, Feb. 2001.

\bibitem{MoNaKoFrPoPeKoHuKu08}
P.~Mogensen, W.~Na, I.~Z.~Kov{\'a}cs, F.~Frederiksen, A.~Pokhariyal, K.~I.~Pedersen,
T.~Kolding, K.~Hugl, and M.~Kuusela.
\newblock LTE capacity compared to the Shannon bound.
\newblock {\em Proc.\ of IEEE VTC Spring '07}, 2007.

\bibitem{NaDo06}
M.~Nawrocki, H.~Aghvami, and M.~Dohler.
\emph{Understanding UMTS Radio Network Modelling, Planning and Automated Optimisation: Theory and Practice.}
Wiley, 2006.

\bibitem{PaHwCh07}
J.~Park, S.~Hwang, and H.-S.~Cho.
\newblock A packet scheduling scheme to support real-time traffic in OFDMA systems.
\newblock {\em Proc.\ of IEEE VTC Spring '07}, 2007.

\bibitem{PoKoMo06}
A.~Pokhariyal, T.~E.~Kolding, and P.~E.~Mogensen.
\newblock Performance of downlink frequency domain packet scheduling for the UTRAN long
term evolution.
\newblock {\em Proc.\ of IEEE PIMRC '06}, 2006.

\bibitem{ShAnEv05}
Z.~Shen, J.~G.~Andrews, and B.~L.~Evans.
\newblock Adaptive resource allocation in multiuser OFDM
systems with proportional rate constraints.
\newblock {\em IEEE Transactions on Wireless Communications}, 4:2726--2737, 2005.

\bibitem{SiFuFo09}
I.~Siomina, A.~Furusk{\"a}r, and G.~Fodor.
\newblock A mathematical framework for statistical {QoS} and capacity studies
  in OFDM networks.
\newblock {\em Proc.\ of IEEE PIMRC '09}, 2009.

\bibitem{SiVaYu06}
I.~Siomina, P.~V{\"a}rbrand, and D.~Yuan.
\newblock Automated optimization of service coverage and base station
antenna configuration in UMTS networks.
\newblock {\em IEEE Wireless Communications Magazine}, 13:16--25, 2006.

\bibitem{SiYu04}
I.~Siomina and D.~Yuan.
\newblock Optimization of pilot power for load balancing in WCDMA networks.
\newblock {\em Proc.\ of IEEE GLOBECOM '04}, 2004.

\bibitem{SiYu12}
I.~Siomina and D.~Yuan.
\newblock Load balancing in heterogeneous LTE: Range optimization via cell offset and load-coupling characterization.
\newblock {\em To appear in Proc.\ of IEEE ICC '12}, June 2012.

\bibitem{SoSh10}
L.~Song and J.~Shen, editors.
\newblock {\em Evolved Cellular network Planning and Optimization for UMTS and LTE}.
\newblock CRC Press, 2010.

\bibitem{StWiBo07}
S.~Sta{\'n}czak, M.~Wiczanowski, and H.~Boche.
\newblock Distributed utility-based power control: objectives and algorithms.
\newblock{\em IEEE Transactions on Signal Processing}, 55:4053--4068, Oct.~2007.

\bibitem{StFeWiBo10}
S.~Sta{\'n}czak, A.~Feistel, M.~Wiczanowski, and H.~Boche.
\newblock Utility-based power control with QoS support.
\newblock{\em Wireless Networks}, 16:1691--1705, 2010.

\bibitem{WiBoSt09}
M.~Wiczanowski, H.~Boche, and S.~Sta{\'n}czak.
\newblock An algorithm for optimal resource allocation in cellular networks with elastic traffic.
\newblock{\em IEEE Transactions on Communications}, 57:41--44, Jan.~2009.

\bibitem{WiStBo08}
M.~Wiczanowski, S.~Sta{\'n}czak, and H.~Boche.
\newblock Providing quadratic convergence of decentralized power control in wireless
networks -- the method of min-max functions.
\newblock{\em IEEE Transactions on Signal Processing}, 56:4053--4068, Aug.~2008.

\bibitem{Za92a}
J.~Zander.
\newblock Performance of optimum transmitter power control in cellular radio systems.
\newblock{\em IEEE Transactions on Vehicular Technology}, 41:57--62, Feb.~1992.

\bibitem{Za92b}
J.~Zander.
\newblock Distributed co-channel interference control in cellular radio systems.
\newblock{\em IEEE Transactions on Vehicular Technology.}, 41:305--311, Aug.~1992.

\bibitem{ZhYaAyWu07}
J.~Zhang, J.~Yang, M.~E.~Aydin, and J.~Y.~Wu.
\newblock UMTS base station location planning: a mathematical model and
heuristic optimisation algorithms.
\newblock {\em IET Communications} 1:1007--1014, 2007.

\bibitem{ZhKoTiGo08}
I.~Zhao, M.~Joonert, A.~Timm-Giel, and C.~G{\"o}rg.
\newblock Highly efficient simulation approach for the network planning of HSUPA in UMTS.
\newblock {\em Proc.\ of IEEE VTC Spring '08}, 2008.

\end{thebibliography}
\end{document}